\documentclass[runningheads]{llncs}
\usepackage{marvosym}
\usepackage{tikz}
\usetikzlibrary{arrows,shapes,automata, calc, chains, matrix, positioning, scopes, patterns, trees}
\usepackage{amsmath}
\usepackage{amssymb}
\usepackage{marvosym}
\usepackage{thmtools}
\usepackage{mathtools}
\usepackage{float}
\usepackage{caption}
\usepackage[full]{complexity}
\usepackage{marginnote}
\usepackage{ stmaryrd }
\usepackage{orcidlink}
\usepackage{comment}
\usepackage{algorithm}
\usepackage{algpseudocode}
\usepackage{enumitem}
\usetikzlibrary{trees}
\usepackage{ stmaryrd }
\usepackage{fourier}
\usepackage{complexity}
\usepackage[capitalise]{cleveref}
\usepackage{xargs}                      
\usepackage[colorinlistoftodos,prependcaption,textsize=tiny]{todonotes}
\usepackage{subcaption}
\usepackage{thm-restate}

\newcommand{\gamegraph}{\mathfrak{G}}

\newcommand{\labelset}{\mathbb{L}}
\newcommand{\words}{\mathbb{W}}

\newcommand{\Controller}{Controller}
\newcommand{\Environment}{Environment}

\newcommand{\seq}[1]{\left\langle #1 \right\rangle}
\newcommand{\tpl}[1]{\left( #1 \right)}
\newcommand{\eset}[1]{\left\{\, #1 \,\right\}}
\newcommand{\set}[1]{\left\{\, #1 \,\right\}}

\newcommand{\mynext}{\mathsf{next}}
\newcommand{\colouring}{\mathsf{colour}}

\newcommand{\accumulatedcolourset}{\mathsf{ColourSet}}

\newcommand{\ancestor}{\mathsf{GCA}}
\newcommand{\parent}{\mathrm{parent}}
\newcommand{\lift}{\mathrm{lift}}
\newcommand{\biglift}{\mathrm{Lift}}

\usepackage{wrapfig}

\newcommand{\Gone}{$\mathsf{G}_\succ$}
\newcommand{\Gtwo}{$\mathsf{G}_\downarrow$}
\newcommand{\Rrr}{$\mathsf{B}$}

\newcommand{\Nb}{\mathbb{N}}

\newcommand{\Lc}{\mathcal{L}}

\newcommand{\Dc}{\mathcal{D}}

\newcommand{\Gc}{\mathcal{G}}

\newcommand{\Tc}{\mathcal{T}}

\newcommand{\Uc}{\mathcal{U}}

\newcommand{\nextbit}[1]{ \mathsf{nextstring}^{#1} }
\newcommand{\strtop}{\lightning}

\newcommand{\nexxt}[1]{ $\mathsf{next}_C^\ell(#1)$ }

\newcommandx{\thejaswiniside}[2][1=]{\todo[linecolor=blue,backgroundcolor=blue!25,bordercolor=blue,#1]{Thejaswini: #2}}
\newcommandx{\irmakinside}[2][1=]{\todo[linecolor=green,backgroundcolor=green!35,bordercolor=green,#1]{Irmak: #2}}

\newcommand{\gr}{\tikz\draw[yellow,fill=yellow] circle (.5ex);}
\newcommand{\bl}{\tikz\draw[blue,fill=blue] circle (.5ex);}
\newcommand{\rd}{\tikz\draw[red,fill=red] circle (.5ex);}
\newcommand{\blck}{\tikz\draw[black,fill=black] circle (.5ex);}

\newcommand{\nocolournode}{\tikz\draw[black,fill=white] circle (.5ex);}
\usepackage [pagewise]{lineno}

\begin{document}
\title{Rabin Games and Colourful Universal Trees\thanks{This
work is a part of the project VAMOS that has received funding from the European Research Council (ERC) under the European Union’s Horizon 2020 research
and innovation programme, grant agreements No 101020093. Rupak Majumdar was partially supported by the DFG project 389792660 TRR 248—CPEC.}\thanks{The full version of the paper~\cite{MST24} is available on arXiv at \url{http://arxiv.org/abs/2401.07548}}}

\author{Rupak Majumdar\inst{1}\orcidlink{0000-0003-2136-0542} \and
Irmak Sa\u{g}lam\inst{1} \and
K. S. Thejaswini\inst{2,3}\orcidlink{0000-0001-6077-7514}}

\authorrunning{R. Majumdar, I. Sa\u{g}lam, K. S. Thejaswini} % First names are abbreviated in the running head.
 \institute{Max Planck Institute for Software Systems (MPI-SWS), Kaiserslautern  
 \email{\{rupak,isaglam\}@mpi-sws.org}\\\and
Department of Computer Science, University of Warwick, UK
 %\url{http://www.springer.com/gp/computer-science/lncs} 
 \and
 Institute of Science and Technology ,Austria\\
 \email{thejaswini.k.s@ista.ac.at}
 }
\maketitle              % typeset the header of the contribution
\begin{abstract}
We provide an algorithm to solve Rabin and Streett games over graphs with $n$ vertices, $m$ edges, and $k$ colours 
that runs in $\Tilde{O}\left(mn(k!)^{1+o(1)} \right)$ time and
$O(nk\log k \log n)$ space, where $\Tilde{O}$ hides poly-logarithmic factors.
Our algorithm is an improvement by a super quadratic dependence on $k!$ from the currently best known run time of
$O\left(mn^2(k!)^{2+o(1)}\right)$, obtained by converting a Rabin game into a parity game, while simultaneously improving its exponential space requirement.

Our main technical ingredient is a characterisation of progress measures for Rabin games using  \emph{colourful trees} and
a combinatorial construction of succinctly-represented, \emph{universal} colourful trees. 
Colourful universal trees are generalisations of universal trees used by Jurdzi\'{n}ski and Lazi\'{c} (2017)
to solve parity games, as well as of Rabin progress measures of Klarlund and Kozen (1991).
Our algorithm for Rabin games is a progress measure lifting algorithm where the lifting is performed on succinct, colourful, universal trees.

% We show a similar characterisation of progress measures for almost sure winning in turn-based stochastic Rabin games.
% This again gives a 
% $\Tilde{O}(mn(k!)^{1+ o(1)})$-time and $O(nk\log k \log n)$-space algorithm, again improving the state of the art.

\end{abstract}

\section{Introduction}

A \emph{Rabin game} is a two-player infinite-duration game played on a directed, coloured graph, 
where each vertex has a finite set of \emph{good} colours and a finite set of \emph{bad} colours associated with it~\cite{Rab69}.
The two players \Controller~and \Environment~take turns to move a token along an edge to form a \emph{play}, an infinite path in the graph. Such a play 
is winning for \Controller~if there is a colour that is a  good colour for some vertex seen infinitely often 
 along the path and is not a bad colour for any vertex seen infinitely often.
Rabin games lie at the core of reactive synthesis for omega-regular specifications
and efficient algorithms for Rabin games are of practical interest in synthesis tools.

% Let us put our new algorithm in historical perspective.
Rabin automata already appear in McNaughton's solution of Church's synthesis problem \cite{Chu57,McN66}
and in Rabin's proof of the decidability of SnS \cite{Rab69}, where it was first defined in the setting of infinite trees.
To solve Church's synthesis problem for $\omega$-regular specifications, represented by non-deterministic B\"{u}chi automata, there are two well-studied (polynomial-time equivalent) approaches: 
either reduce it to the emptiness problem for Rabin tree automata or solve a Rabin game.

Rabin conditions are also suitable specifications for \emph{general fairness constraints}~\cite{FK84}. 
Klarlund and Kozen~\cite{KK91} defined Rabin measures over graphs and applied them to prove program termination under a general fairness constraint. Indeed, the acceptance condition that defines \emph{strong fairness}, i.e., if a given set of actions (edges) is enabled infinitely often (the source vertex is seen infinitely often), it is taken infinitely often, is naturally expressed by the complement of the Rabin condition, called the Streett condition~\cite{Str81}. 

Algorithmically, the problem of solving Rabin games was shown to be $\NP$-complete by Emerson and Jutla~\cite{EJ88,EJ99} in the late 1980s. In the same paper, Emerson and Jutla, and independently, Pnueli and Rosner \cite{PR89}, gave an algorithm that takes time $O\left((nk)^{3k}\right)$ time, where $n$ is the number of vertices of the game graph and $k$ the number of colours. 

Steady progress was made to solve Rabin games, and within a decade, 
Kupferman and Vardi~\cite{KW98} reduced the cubic dependence on $n^k$ to a quadratic one by giving an algorithm to check non-emptiness in a Rabin tree automata in time $O\left(mn^{2k}k!\right)$. 
Later, Horn~\cite{Hor05} gave a different solution to solve Streett games---and therefore Rabin games---with the same running time.  

A lot of progress was simultaneously made on \emph{parity games}~\cite{EJ91}, a special case of Rabin games where colours are assigned to each subset of states in a chain of subsets.
%
% Such a translation converts a Rabin game on $n$ vertices with $k$ colours to a parity game on $N = nk^2k!$ many vertices, $M = Nm$ many edges, and $K = 2k+1$ many priorities as made explicit in the work of Emerson and Jutla~\cite{EJ91}, obtained from a slight modification of the arguments in the thesis of Safra~\cite{Saf88, Saf89} and Gurevich and Harrington~\cite{GH82}. 
%\thejaswini{PP06 is not  based on the Parity-Rabin reduction. They kind of bu-pass the reduction and just do a lifting algorithm inspired by the }
%Using this reduction, and 
Inspired by fixpoint evaluation algorithms~\cite{EJ91} and the small progress measure algorithm~\cite{Jur00} of Jurdzi\'{n}ski for parity games, Piterman and Pnueli~\cite{PP06} gave a fast $O\left(m n^{k+1}k k!\right)$-time, $O(nk)$-space, algorithm for Rabin games. 
This algorithm used a concept of a measure to solve Rabin games. 

The work of Piterman and Pnueli remained state-of-the-art for Rabin games until the quasi-polynomial breakthrough for parity games by Calude, Jain, Khoussainov, Li, and Stephan~\cite{CJKLS22}. 
They gave a fixed parameter tractable algorithm (FPT) for Rabin games on $k$ colours 
%where the dependence on the number of colours is of the order $k^{5k}$. %It is known that a Rabin game on $k$ colours can be translated to a Muller game on at most $2k$ colours, thereby giving a direct algorithm to solve Rabin and Street games in time proportional to $O((2k)^{10k})$.
%They further remarked that the best way to solve a Rabin game is 
by converting it to a parity game and using the quasi-polynomial algorithm.
 %or alternatively the algorithm of Fearnley et al.~\cite{FJKSSW19}, which provides a quasi-bi-linear running time of  $O(MN\log(N)^{K-1})$.

A Rabin game with $n$ vertices, $m$ edges, and $k$ colours, can be reduced to a parity game over $N= nk^2k!$ vertices, $M = nk^2k!m$ edges, and $K = 2k+1$ colours \cite{EJ91}.
By combining the reduction from Rabin to parity games and state-of-the-art algorithms for parity games~\cite{JL17,DJT20,FJKSSW19,DE22} in a ``space-efficient'' manner, 
say of Jurdzi\'{n}ski and Lazi\'{c}~\cite{JL17}, one can solve Rabin games in time $O\left(\max\left\{MN^{2.38}, 2^{O(K\log K)}\right\}\right)$,
but exponential space (since the parity game is exponentially bigger).

%Although neither of these bounds stated immediately attractive to solve Rabin games, f
%\RM{not sure I understand the math here}
On substitution of the values of $M$ and $N$, the algorithm of Jurdzi\'{n}ski and Lazi\'{c}  would take time at least  proportional to 
$m(nk^2\cdot k!)^{3.38}$ for games with $n$ vertices, $m$ edges and $k$ colours. 
 However, observe that the parity game obtained from a Rabin game is such that the number of vertices $N=nk^2k!$ is much larger than the number of colours $K=2k+1$. Indeed, this results in $K\in o\left(\log(N)\right)$. 
For  cases where the number of vertices of the resulting parity game is much larger than the number of priorities, say the number of colours ($2k+1$) is $o\left(\log\left(N\right)\right)$---which is the case above as $k$ grows---Jurdzi\'{n}ski and Lazi\'{c} also give an analysis of their algorithm that would  solve Rabin games in time $O\left(nmk!^{2+o(1)}\right)$. %\irmakinside{suggestion for previous sentence due to review: ...--...  } 
Closely matching this 
are the run times in the work of Fearnley et al.~\cite{FJKSSW19} who provide, among other bounds, a quasi-bi-linear bound of $O\left(MN\mathfrak{a}(N)^{\log{\log {N}}}\right)$, 
where $\mathfrak{a}$ is the inverse-Ackermann function. 
In either case above, this best-known algorithm has at least a $(k!)^{2+o(1)}$ dependence 
in its run time, and takes space proportional to $(nk^2k!)\log (n k^2k!)$, which has a $k!$ dependence again. 

\paragraph{Our Contribution.} 
Our result breaks through the $2+o(1)$ barrier, while simultaneously using polynomial space, 
to give a fixed-parameter tractable algorithm for Rabin games. % where the parameter is the number of colours in the game.
We show a new algorithm for Rabin games on graphs that runs in time $\Tilde{O}(mn(k!)^{1+o(1)})$
time and $O(nk \log k \log n)$ space, for a game on $n$ vertices, $m$ edges, and $k$ colours.
Our algorithm improves the quadratic $(k!)^2$ dependence in the number of colours 
in the best current algorithms, while simultaneously using only polynomial space.

Our first technical contribution is a characterisation of winning states in Rabin games using ``\emph{colourful trees}'',
by generalizing previous work on Rabin measures on graphs by Klarlund and Kozen~\cite{KK91}. 
Using our characterisation, we provide an algorithm to compute winning states and strategies as a fixed point
of a lifting function over the lattice of functions from vertices of a game to nodes of a colourful tree. 

Our second contribution is the construction of a \emph{universal} colourful tree that embeds any colourful tree with a given number of leaves and fixed set of colours. Universal trees are found underlying all the quasi-polynomial algorithms for parity games~\cite{JL17,CDFJLP19,JMT22,DJT20,KL22}.
Our construction uses the theory of universal trees developed for parity games, especially that of Jurdzi\'{n}ski and Lazi\'{c}~\cite{JL17}. From our construction of universal colourful trees, we can also naturally construct an instance of universal graphs for Rabin objectives, where the definition of universal graph is as introduced by Colcombet and Fijalkow. Although constructing universal graphs directly give us a lifting algorithm, for the sake of completeness, we also provide a lifting algorithm that uses our construction of colourful universal trees. Therefore, we show how to construct a small universal colourful tree (our upper bound is tight up to a polynomial factor) that can be succinctly encoded and efficiently navigated.

By applying the lifting algorithm to our succinct universal colourful tree, we get our time and space bounds.

Just as Piterman and Pnueli's result generalized ranking techniques and progress measures for parity games, we generalize the notion of measures~\cite{KK91} and universal trees~\cite{JL17} central to the fastest algorithms for parity games to obtain
our algorithm.

% Our ideas can be generalized to the setting of Rabin games with \emph{transition fairness}~\cite{QS83,BMMSS22} and 
% turn-based stochastic Rabin games~\cite{CAH05}.
% We provide a new definition of Rabin measures for transition-fair Rabin games using colourful trees.
% Together with our lifting algorithm and universal colourful trees, we obtain a 
% $\widetilde{O}(mn(k!)^{1+ o(1)})$-time and $O(nk\log k \log n)$-space algorithm for transition-fair Rabin games 
% as well as for almost sure winning in turn-based stochastic Rabin games.
% While there is a known reduction from transition-fair Rabin games to (usual) Rabin games with at most $nk$ 
% our algorithm shaves off a $k^2$ factor from the above in the worst case.

\section{Preliminaries}
We use $\mathbb{N}$ to denote the set of all natural numbers $\{0,1,2,\dots\}$. 
A directed graph consists of a finite set of vertices $V$ along with a binary relation $E$ over the set of vertices called the \emph{edge set}. 
We write $u\rightarrow v$ to denote an edge $(u,v)\in E$. 
A finite (resp. infinite) \emph{path} in a directed graph is a finite (resp. infinite) sequence of vertices such that a tuple 
formed by any two consecutive vertices in this sequence is an edge in $E$. 
%Given an infinite path $\pi$, we denote with $\inf(\pi)$, the set of vertices that occur infinitely often in this path. 
%\RM{check: do we use $\inf(\pi)$ later?}

\paragraph*{$(c_0,C)$-Colourful Ordered Trees.} 
Let $C$ be a finite set of colours and let $c_0\notin C$ be a distinguished \emph{root colour}.
Informally, a $(c_0,C)$-colourful ordered tree with root colour $c_0$ is an ordered tree
of height at most $|C|+1$ whose 
root is associated with the colour $c_0\notin C$, and whose every other node has a colour from $C$ associated to it. 
As an exception, we allow some leaves to be left uncoloured, denoted by a ``dummy colour'' $\bot\notin C$. 
We also require that along any path from the root to a leaf, each node  must have a different colour.

Formally, for a finite set $C$, we recursively define $(c_0,C)$-colourful trees
\begin{itemize}
    \item  if $C = \emptyset$, $(c_0,\seq{})$ and $\tpl{c_0,\seq{\left(\bot,\seq{}\right),\dots,\left(\bot,\seq{}\right)}}$ are  $(c_0,\emptyset)$-colourful trees.
    \item if $C\neq \emptyset$, we say $\Tc$ is $(c_0,C)$-colourful tree if it is either 
    \begin{itemize}
        \item a $(c_0,C')$-colourful tree rooted at $c_0$ for some $C'\subsetneq C$; or
        \item $\Tc = \tpl{c_0,\seq{\Tc_1,\dots,\Tc_\ell}}$, and for all $i\in\set{1,\ldots,\ell}$, either there is a $c_i\in C$ and $\Tc_i$ is 
        a $(c_i, C\setminus \{c_i\})$-colourful ordered tree, or  $\Tc_i = (\bot,\seq{})$. 
        Note that these $c_i$ need not be different from one another.
    \end{itemize} 
    
\end{itemize}
We define the \emph{concatenation} of a $(c_0, C_1)$-colourful tree $\Tc_1 = \tpl{c_0,\seq{\Tc^1_1,\dots,\Tc^m_1}}$ and a $(c_0, C_2)$-colourful tree $\Tc_2 = \tpl{c_0,\seq{\Tc^1_2,\dots,\Tc^\ell_2}}$
as the $(c_0, C_1\cup C_2)$-colourful tree denoted by $\Tc_1\cdot \Tc_2$ as $\tpl{c_0,\seq{\Tc^1_1,\dots,\Tc^m_1,\Tc^1_2,\dots,\Tc^\ell_2}}$.
For a root colour $c_0$, a number $\ell\in \Nb$, and a $(c,C)$-colourful ordered tree $\Tc$, we denote $\Tc^\ell$ to be the tree with $\ell$ many copies of $\Tc$, $\tpl{c_0,\seq{\Tc, \Tc,\dots,\Tc}}$.
% where the root colour is assumed to be $c_0$ by default unless stated otherwise. 
% In general, whenever we do not explicitly mention the colour a tree is rooted at, it is assumed to be $c_0$. 
When $(c_0, C)$ is clear from context, we simply say ``colourful tree.''

\paragraph*{Embedding Colourful Trees.}
Given a $(c_0,C)$-colourful tree $\Uc$ and a $(c_0,C')$-colourful tree $\Tc$, such that  $C'\subseteq C$, we say $\Uc$ \emph{embeds} $\Tc$ if 
 $\Tc =(c_0,\seq{})$, or 
 $\Tc =  \tpl{c_0,\seq{\Tc_1,\dots,\Tc_\ell}}$ and $\Uc =  \tpl{c_0,\seq{\Uc_1,\dots,\Uc_m}}$ for some $\ell,m$, and there is some increasing sequence of indices $1\leq i_1<i_2<\dots< i_\ell\leq m$ such that 
     $\Uc_{i_j}$ embeds $\Tc_{j}$ recursively.
    Notice both $\Uc_{i_j}$ and $\Tc_{j}$ must be rooted at the same colour, say $c_j$ and both are $\tpl{c_j,C\setminus \{c_j\}}$-colourful and $\tpl{c_j,C'\setminus\{c_j\}}$-colourful trees respectively.

\paragraph*{Labelled Colourful Trees.}
In what follows, we shall additionally label colourful trees with labels from some linearly ordered set.
It is more convenient to define such labelled colourful trees as prefix-closed sets of sequences, using the isomorphism
between a (recursively defined) tree and its set of paths.

Let $\labelset$ be a set of labels with a linear ordering $<_\labelset \subseteq \labelset \times \labelset$. 
An $\labelset$-labelled $(c_0, C)$-colourful tree is a prefix-closed set of sequences over $\labelset\times (C\cup\{\bot\})$ 
 where $\labelset\times (C\cup\{\bot\})$ is the Cartesian product of $\labelset$ and $(C\cup\{\bot\})$.

Given an element $\tau_0\in \labelset\times (C\cup\{\bot\})$ and a sequence $\tpl{\tau_1,\tau_2,\dots,\tau_j}$ in $\tpl{\labelset\times (C\cup\{\bot\})}^*$, we use $\odot$ to denote concatenation to the tuple, where we say $\tau_0 \odot \tpl{\tau_1,\tau_2,\dots \tau_j} \: = (\tau_0,\tau_1,\tau_2,\dots \tau_j)$. We extend this notation to sets of sequences $\Lc$, by also defining
$\tau_0\odot\Lc = \eset{ (\tau_0,\tau_1,\tau_2,\dots \tau_j) \mid (\tau_1,\tau_2,\dots \tau_j)\in \Lc}$.

We say % that an $\labelset$-labelled $(c_0,C)$-colourful tree 
a prefix-closed set $\Lc\subseteq (\labelset \times (C\cup \set{\bot}))^*$ is an $\labelset$-labelling of a $(c_0,C)$-colourful ordered tree $\Tc$
\begin{itemize}
    \item if $\Tc = \tpl{c_0,\seq{\tpl{\bot,\seq{}}^m}}$, and $\Lc$ is the prefix closure of the set $\set{(\alpha_1,\bot),\dots, (\alpha_m,\bot)}$ 
    for some $\alpha_1<_\labelset \alpha_2<_\labelset \dots<_\labelset \alpha_m\in \labelset$, 
    \item if $\Tc = \tpl{c_0, \seq{\Tc_1,\dots,\Tc_m}}$  then 
    $\Lc$ is the prefix closure of the set 
    $$(\alpha_1, c_{1})\odot\Lc_1\cup (\alpha_2, c_{2})\odot\Lc_2 \cup \dots \cup (\alpha_m, c_{m})\odot\Lc_m$$
    for some $\alpha_1\leq_\labelset  \alpha_2\leq_\labelset \dots\leq_\labelset  \alpha_m$ in $\labelset$, such that for all $j$,
    \begin{itemize}
        \item $\Tc_j$ is a $C\setminus\{c_{j}\}$-colourful tree rooted at $c_j$ and $\Lc_{j}$ is an $\labelset$-labeling of $\Tc_j$, 
        \item $c_{j}\in C\cup\{\bot\}$, and
        \item whenever $\alpha_j = \alpha_{j+1}$, we have $c_{j}\neq c_{j+1}$
    \end{itemize}
\end{itemize}
Note that the root colour $c_0$ of $\Tc$ does not appear in $\Lc$; instead of tracking $c_0$ along with $\Lc$ explicitly, 
we implicitly assume the root colour of the tree $\Lc$ above is $c_0$. %\irmakinside{Sugg 2: See Fig.~\ref{fig:labelledtree} for an example of a labelled colourful tree.}

We refer to elements of the prefix-closed set $\Lc$ of a labelled tree as \emph{nodes} of the tree. 
For two nodes $n_1$ and $n_2$ in $\Lc$,
% the  $\labelset$-labelled $(c_0,C)$-colourful tree $\Lc$, 
we define \emph{the greatest common ancestor}, written $\ancestor(n_1,n_2)$, as the longest common prefix  of $n_1$ and $n_2$. 
We define $n_1$ to be an \emph{ancestor} of $n_2$ if $n_1 = \ancestor(n_1,n_2)$. 
In particular, $n_1$ is a parent of $n_2$, written $n_1 = \parent(n_2)$, if $n_1$ is the largest node other than $n_2$ such that $n_1 = \ancestor(n_1,n_2)$;
we then say $n_2$ is a child of $n_1$.

The colouring of a node is defined to be the last colour occurring in the sequence: 
For the empty sequence $\tpl{}$, we define $\colouring(\tpl{}) = c_0$, and  $\colouring((\alpha_1, c_{i_1}),\dots,(\alpha_j, c_{i_j})) = c_{i_j}$.
Furthermore we define $\accumulatedcolourset: \Lc \rightarrow 2^{C\cup\set{c_0}}$, 
which maps a node to the set of colours seen from the root to that node:
$\accumulatedcolourset(n) = \set{\colouring(n')\mid n' = \ancestor(n', n)} \setminus \set{\bot}$.

\paragraph*{Ordering.}
We define an ordering $\prec_\Lc$ on $\Lc$.
First, we fix some arbitrary linear order on the set $C$ and set colour $\bot$ to be larger than all the colours in $C$ in the ordering. 
We compare elements by extending the linear order $<_{\labelset}$ over $\labelset$ and an arbitrary fixed order $<$ over $C$ %on $\labelset$ 
to a linear order over the set $\labelset\times (C\cup\{\bot\})$ lexicographically as follows: for two elements in $\labelset\times (C\cup\{\bot\})$, we declare $(\alpha_1,c_1) < (\alpha_2,c_2)$ if either $\alpha_1 <_{\labelset} \alpha_2$ or $\alpha_1 = \alpha_2$ and $c_1<c_2$.
%$\labelset\times (C\cup\{\bot\})$ lexicographically. \RM{check: it was not clear to me what the natural ordering was.}

For two nodes $n_1, n_2\in \Lc$, we define $n_1\prec_\Lc n_2$ if either $n_1$ is a strict prefix of $n_2$, or 
if $n_1$ is lexicographically smaller than $n_2$ when viewed as sequences over $\labelset \times (C\cup\set{\bot})$. 

\begin{example}
    Figure~\ref{fig:colourfultree} depicts a $(\blck, \left\{ \gr,\rd,\bl \right\})$-colourful tree, where the nodes denoted by $\nocolournode$~represents uncoloured nodes. A fixed ordering on the set of colours $\gr<\bl<\rd<\nocolournode$, a labelling of this tree over $\labelset = \{1, 2\} \subseteq \mathbb{N}$ is the prefix closure of the set 
$\{ (1 \,\gr, 1\, \bl, 1\,\nocolournode )$, 
$(1 \,\gr, 1\, \bl, 2\, \rd,  1\,\nocolournode )$, 
$ (1 \,\gr, 1\, \bl, 2\, \rd,  2\,\nocolournode )$, 
$(1 \,\gr, 1\, \rd, 1\, \bl)$, 
$(1 \,\gr, 1\, \rd, 2\, \bl, 2\,\nocolournode)$, 
$(1\,\nocolournode)$, 
$(2\,\bl, 2\,\gr)$, 
$(2\,\rd, 1\,\bl, 1\,\gr, 1\,\nocolournode)$, 
$(2\,\rd, 1\,\bl, 2\,\nocolournode)\}$. The ordering $\prec_\Lc$, (represented by $\prec$) on some nodes is as follows: $()\prec(1\gr)\prec (1\gr,1\rd)\prec (1\nocolournode) \prec(2\bl,2\gr)$. The  ordering in the nodes of the tree in the figure decreases when we go from a child to a parent, or we go ``left'' in the tree, but otherwise increases.
\end{example}
%For $\Lc$ being the tree in Figure~\ref{fig:colourfultree}, the order with respect to $\prec_\Lc$  (henceforth referred to by $\prec$, rather than $\prec_\Lc$) is as follows on the following nodes $()\prec(1\gr)\prec (1\gr,1\rd)\prec (1\nocolournode) \prec(2\bl,2\gr)$. The  ordering on a tree
%decreases when we go from a child to a parent, or we go ``left'' in the tree, but otherwise increases.

\begin{figure}
\centering
\begin{minipage}{.5\textwidth}
  \centering
    \begin{tikzpicture}[nodes={draw,draw=white!0=10}]
\node[draw] at (-0.7,-0.1) {$1$};
\node[draw] at (-0.4,-0.4) {$1$};
\node[draw] at (-1,-1) {$1$};
\node[draw] at (-2,-1) {$1$};
\node[draw] at (-2.6,-1.7) {$1$};
\node[draw] at (-1.8,-1.7) {$2$};
\node[draw] at (-1.2,-1.7) {$1$};
\node[draw] at (-0.4,-1.7) {$2$};
\node[draw] at (-0.3,-2.5) {$2$};
\node[draw] at (-1.5,-2.5) {$2$};
\node[draw] at (-2.2,-2.5) {$1$};
\node[draw] at (0.7,-0.1) {$2$};
\node[draw] at (0.4,-0.4) {$2$};
\node[draw] at (0.7,-1.1) {$2$};
\node[draw] at (1.7,-1.1) {$2$};
\node[draw] at (1.2,-1.7) {$1$};
\node[draw] at (1.8,-1.7) {$1$};
\node[draw] at (1.3,-2.6) {$2$};
\node at (0,0) {$\blck$}
        child[level distance=8mm, sibling distance = 10mm] { node{$\gr$}
            child[level distance=7mm, sibling distance = 14mm] { node {$\bl$} 
                child[level distance=7mm, sibling distance = 7mm] { node {$\nocolournode$} }
                child[level distance=7mm, sibling distance = 7mm] { node {$\rd$} 
                    child[level distance=7mm, sibling distance = 7mm] { node {$\nocolournode$}  }
                    child[level distance=7mm, sibling distance = 7mm] { node  {$\nocolournode$} }
                }
            }
            child[level distance=7mm, sibling distance =14mm] { node {$\rd$} 
                child[level distance=7mm, sibling distance =7mm] { node {$\bl$} }
                child[level distance=7mm, sibling distance =7mm] { node {$\bl$} 
                    child[level distance=7mm] { node {$\nocolournode$} }
                }
            }
        }
        child[level distance=8mm, sibling distance =10mm] { node{$\nocolournode$} }
        child[level distance=8mm, sibling distance =10mm] { node{$\bl$}
            child[level distance=7mm] { node {$\gr$} }
        }
        child[level distance=8mm, sibling distance =10mm] { node {$\rd$} 
                child[level distance=7mm, sibling distance =7mm] { node {$\bl$} 
                        child[level distance=7mm] { node {$\gr$} %node [label=right:{\scriptsize {16}}]
                            child[level distance=7mm] { node {$\nocolournode$} }
                        }
                        child[level distance=7mm, sibling distance =7mm] { node {$\nocolournode$} }
                }
        };
\end{tikzpicture}    
  \captionof{figure}{A colourful tree.}
    \label{fig:colourfultree}
\end{minipage}%
\begin{minipage}{.5\textwidth}
  \centering
  \makebox[\textwidth][c]{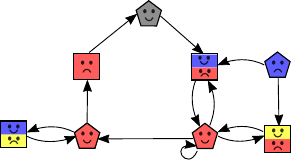} 
  \captionof{figure}{A colourful Rabin graph $\Gc$ where all infinite paths satisfy the Rabin condition}
  \label{fig:RGExample}
\end{minipage}
\end{figure}

 \section{Rabin measure and Colourful Decompositions }~\label{sec:shapeofRG}
 In this section, our aim is to understand the Rabin acceptance condition on graphs. 
We define such acceptance conditions and provide a local witness called a \emph{Rabin measure} for graphs where all paths satisfy the Rabin condition. 

A $(c_0,C)$-colourful \emph{Rabin graph} $\Gc$  
%$= (V, E, \set{G_v}_{v\in V}, \set{B_v}_{v\in V})$ 
consists of 
(1)~a directed graph $(V,E)$,
(2)~a finite set $C$ of \emph{colours} and a special colour $c_0\notin C$, and 
(3)~for each vertex $v\in V$, a set of \emph{good} colours $G_v \subseteq C\cup\eset{c_0}$ for $v$ 
and a set of \emph{bad} colours $B_v \subseteq C$ for $v$. 
Observe that $c_0\notin B_v$ for any $v$.
We call each colour $c$ in $G_v$ a \emph{good} colour for $v$, and each colour in $B_v$ a \emph{bad} colour for $v$.

We assume every vertex has some outgoing edge in the directed graph.
 An infinite path in $\Gc$ satisfies the \emph{Rabin condition} if there is some colour $c$ in $C\cup \{c_0\}$ such that 
 $c$ is a good colour for some $v$ seen infinitely often along the path
 and $c$ is not a bad colour for any $v$ seen infinitely often along the path. 

\begin{example}
Consider the $(\blck,\{\rd,\bl,\gr\})$-colourful Rabin game in~\cref{fig:RGExample}.
%and its $(\blck,\{\rd,\bl,\gr\})$-colourful decomposition in \cref{fig:RGdecomp}. 
The colours that are in the good set of each vertex are represented with a smiley face in the same colour
and those that are bad colours appear with a sad face.  
Although a vertex can have more than one colour assigned to it as a good colour (or a bad colour), we only consider at most one good and bad colour per vertex for this example. 
In our example, the leftmost vertex in the graph $\Gc$ in~\cref{fig:RGExample} has the singleton set $\{\bl\}$ as the set of good colours and the set $\{\gr\}$ as the set of bad colours. Similarly, the topmost vertex in~\cref{fig:RGExample} has the set $\{\blck\}$ as the set of good colours and an empty set of bad colours. Observe that in the graph $\Gc$, any infinite path satisfies the Rabin condition. Indeed, for any infinite path there is some colour that is not a bad colour for any of the vertices that occur infinitely often and is a good colour for some vertex that occurs infinitely often. For example, if a path is such that all the vertices of $\Gc$ are visited infinitely often, then the colour $\blck$ is not a bad colour of any vertex and the same colour $\blck$ is a good colour of the topmost vertex. 
\end{example}
  As opposed to preexisting definition in literature of Rabin games that use Rabin pairs to represent the acceptance condition, we instead define two sets of colours associated to a vertex rather than a pair of subsets of vertices associated to a colour. This does not add more than a constant factor in terms of representation size.
 \paragraph*{A Measure for Rabin Graphs.}
We fix a $(c_0, C)$-colourful Rabin graph $\Gc$ with the underlying graph $(V,E)$ 
with good colours for a vertex $v$ denoted by $G_v$ and the bad colours denoted by $B_v$.
Let $\labelset$ be a linearly ordered set of labels, and let $\Lc$ be an $\labelset$ labelled $(c_0,C)$ -coloured tree. 
We define $\Lc^\top = \Lc \cup\set{\top}$ by adjoining an element $\top$ to $\Lc$ and
we extend the ordering $\prec_\Lc$ (denoted henceforth by $\prec$) to $\Lc^\top$, by declaring $t\prec\top$ for all $t\in \Lc$. 

Consider a map $\mu: V \rightarrow \Lc^\top$. 
We call an edge $u\rightarrow v$ \emph{consistent} with respect to $\mu$, if either $\mu(u)$ is mapped to $\top$ or it satisfies the condition (\Gone~OR \Gtwo) AND  \Rrr; for \Gone, \Gtwo, and \Rrr~defined below.
\begin{itemize}
    \item[(\Gone)] $\mu(u)\succ\mu(v)$ 
    \item[(\Gtwo)] $\ancestor(\mu(u),\mu(v)) = \mu(u)$ and $\colouring(\mu(u)) \in G_u$.
    \item[(\Rrr)] $\accumulatedcolourset(\mu(u))\cap B_u = \emptyset$
\end{itemize}
In words, \Gone~conveys that the measure $\mu$ decreases along the edge $u\rightarrow v$ and \Gtwo~says that the measure can increase along an edge but only into a descendent node and only when the colour of the node that is currently mapped to is a good colour for $u$.
The condition represented by \Rrr~says that none of the colours assigned to any ancestor of $u$ is a bad colour for it.

% \irmak{Intuitively, the mapping of a consistent path in a Rabin graph correspond to a sequence of nodes in the tree that satisfies the following: the order of the nodes in this sequence either decrease (\Gone), or increase a little bit at the expense of seeing a good colour for the corresponding vertex (\Gtwo) at each step. Furthermore, no bad colours of any vertices are seen on the sequence (\Rrr). Take a winning  $(c_O, C)$-colourful Rabin graph $\Gc$. Then there exists a mapping that sends any cycle in $\Gc$ to such a sequence of tree nodes. The reason intuitively is that, there exists a colour $c$ that is not a bad colour for any of the vertices in the cycle and a good colour for at least one of them (say, $v$). Then we can map $v$ to a node in the tree that it coloured with $c$ and has no ancestors coloured with any bad colours for the vertices in the cycle. Then we assign all the other vertices in the cycle for descendants of this node, in a way that satisfies \Gone.}

If the map $\mu$ is clear from the context, we call an edge or a vertex consistent without mentioning the mapping. 
We say the relation and function $\ancestor(\cdot,\top)$ and $\colouring(\top)$ are undefined, 
and the condition \Gtwo~or \Rrr~are not satisfied when $\mu(v)$ is mapped to $\top$ and $\mu(u)$ is not mapped to $\top$.

We say the map $\mu$ is a $(c_0,C)$-colourful \emph{Rabin measure} for a graph $\Gc$ if all edges in $E$ 
are consistent with respect to $\mu$. A mapping from the vertices of a Rabin graph to the nodes of a tree ensures that an infinite play corresponds to an infinite set of nodes in a tree. If a mapping is consistent, then such a mapping serves as a witness to the fact that an infinite path in the Rabin graph satisfies the Rabin condition.

Our definition is a modification of Klarlund and Kozen's~\cite{KK91} notion of Rabin measures, following recent approaches to faster algorithms for parity games~\cite{JL17,DJT20}.

%\irmak{I thought the below explanation could be a nice intuition but now I am not so sure :/}

\paragraph*{Colourful Decomposition.}
The Rabin measure, as with other progress measures, is based exclusively on local properties. 
Indeed, in the above case, we have a progress measure when each  edge satisfies certain conditions. 
Before we show that Rabin measures capture winning sets of a graph, we define an intermediate structure, which we call \emph{colourful decompositions}. 
These colourful decompositions of a Rabin graph highlight a recursive structure that captures the 
acceptance of all paths in a way which relates naturally to colourful trees.
Colourful decompositions generalise  attractor decompositions of parity games to Rabin games~\cite{DJL19,JMT22,DJT20}.

% \begin{figure}
%   \begin{subfigure}{0.48\textwidth}
%     %\centering
%     \makebox[\textwidth][c]{\input{RGExample.pdf_tex}}%
%     \caption{A colourful Rabin graph $\Gc$ where all infinite paths satisfy the Rabin condition} \label{fig:RGExample}
%   \end{subfigure}%
%   \hspace*{\fill}   % maximize separation between the subfigures
%   \begin{subfigure}{0.43\textwidth}
%     %\centering
%     \makebox[\textwidth][c]{\input{RGDecomp.pdf_tex}}%
%     \caption{A colourful Decomposition of $\Gc$} \label{fig:RGdecomp}
%   \end{subfigure}%
% \caption{A Rabin graph and its decomposition}
% \end{figure}

% \RM{some of the discussion can go to related work}
% We define a colourful decomposition as a structure  that resemble trees and later show that a graph that satisfies Rabin condition on all paths indeed has 
% such a decomposition. These are similar to the objects defined as a structured witness of winning for parity games, called \emph{attractor decompositions} in the work of Jurdzi\'{n}ski and Morvan~\cite{JMT22} and used further in the work of Daviaud, Jurdzi\'{n}ski, and Thejaswini~\cite{DJT20}. 
% Later, we argue that the Rabin measure, which is defined as a map from the set of vertices to a colourful tree that satisfies certain properties is noting but a local view of such colourful decompositions.

Consider a $(c_0,C)$-colourful Rabin graph $\Gc$. %$\Gc = (V, E, \set{G_v}_{v\in V}, \set{B_v}_{v\in V})$.
A $(c_0,C)$-colourful decomposition $\Dc$ of $\Gc$ 
is a recursive sub-division
of vertices $V$ of $\Gc$ into subsets of vertices defined as follows.
 If $C= \emptyset$,  then we say $\Dc\, :=\,\seq{V}$ is a $(c_0,C)$-colourful decomposition if and only if 
 all infinite paths from all vertices in $V$ visit a vertex $v$ such that $c_0\in G_v$.
Else, if $C\neq \emptyset$ and if $|V| \geq 1$, and  
$$\Dc\:= \seq{A, \tpl{c_{1},V_1, \Dc_1, A_1}, \dots, \tpl{c_{j},V_j,\Dc_j, A_j}}$$
satisfies the following conditions:
\begin{itemize}    
    %\item $c$ is a colour such that $c\notin B_v$ for all $v\in V$;
    \item $A$ is the set of all vertices in $V$ such that all infinite paths starting from $A$ in $\Gc$ visit some vertex $v\in V$ such that $c_0\in G_v$;
    \item Set $W_1 = V \setminus A$. For $i\in 1,\dots,j$,
    \begin{itemize}
        \item $V_i$ is a set of vertices which has no path to $W_i\setminus V_i$ and  $c_i\notin B_v$ for all $v\in V_i$;\label{item:decomp-Gc-1}
        \item $\Dc_i$ is a $(c_i,C\setminus\{c_i\})$-colourful decomposition of $V_i$. 
        \item $A_i$ is the set of all vertices in $W_i$ such that all infinite paths from $A_i$ within $W_i$ visits some vertex in $V_i$;
        \item $W_{i+1}\:= W_i\setminus A_i$.
    \end{itemize}
    \item $W_{j+1} =\emptyset$.
\end{itemize}
\begin{figure}
  \begin{subfigure}{0.48\textwidth}
   %\centering
    \makebox[\textwidth][c]{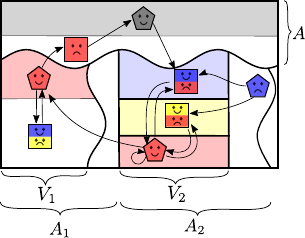}%
    \caption{A colourful decomposition of a Rabin graph $\Gc$ where all paths satisfy the Rabin condition} \label{fig:RGDecomp}
  \end{subfigure}%
  \hspace*{\fill}   % maximize separation between the subfigures
  \begin{subfigure}{0.43\textwidth}
    \centering
    \begin{tikzpicture}[nodes={draw, draw=white!0=10}]
\node[draw] at (-0.7,-0.2) {$0$};
\node[draw] at (-0.4,-0.5) {$1$};
\node[draw] at (-0.3,-1.4) {$0$};

\node[draw] at (0.7,-0.2) {$2$};
\node[draw] at (0.4,-0.5) {$2$};
\node[draw] at (0.7,-1.4) {$1$};
\node[draw] at (0.7,-2.2) {$1$};
\node[draw] at (-0.3,-4.1) {};
\node at (0,0) {$\blck$}
        child[level distance=6mm, sibling distance = 6mm] { node{$\nocolournode$}
        }
        child[level distance=10mm, sibling distance = 10mm] { node{$\rd$}
            child[level distance=8mm] { node {$\nocolournode$} 
            }
        }
        child[level distance=10mm, sibling distance =10mm] { node {$\bl$} 
                child[level distance=8mm, sibling distance =8mm] { node {$\gr$} 
                        child[level distance=8mm] { node {$\rd$} }  
                }
        }
        child[level distance=6mm, sibling distance =6mm] { node {$\nocolournode$} }  ;
\end{tikzpicture}  %
    \caption{A labelled colourful tree into which the graph $\Gc$ has a Rabin measure.} \label{fig:RGmeasure}
  \end{subfigure}%
\caption{A colourful decomposition and tree for Rabin measure}\label{fig:RGimage}
\end{figure}
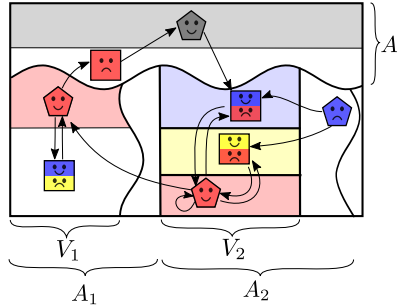
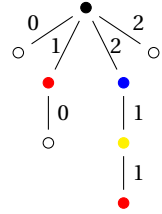

The crux of this section is Theorem~\ref{thm:measure} below which shows the equivalence between Rabin measure, the existence of a colourful decomposition and a Rabin graph where all paths satisfy the Rabin condition. 
\begin{restatable}{theorem}{measure}
\label{thm:measure}
The following three statements are equivalent for a $(c_0,C)$-colourful Rabin graph $\Gc$.
\begin{enumerate}
    \item All infinite paths in $\Gc$
    satisfy the Rabin condition.~\label{en:one}
    \item There is a $(c_0,C)$-colourful decomposition $\Dc$ of the vertices of $\Gc$.~\label{en:decomp}
    \item There is 
 an $\labelset$-labelled $(c_0,C)$-colourful Rabin measure for $\Gc$, where no vertex is mapped to $\top$ for some linearly ordered infinite set $\labelset$.~\label{en:measure} 
\end{enumerate}
\end{restatable}
The theorem above is proved by showing~\ref{en:one}$\implies$\ref{en:decomp} in Lemma~\ref{lemma:winningdecomposition},
\ref{en:decomp}$\implies$\ref{en:measure} in Lemma~\ref{lemma:decompToMeasure} and finally \ref{en:measure}$\implies$\ref{en:one} in Lemma~\ref{lemma:measureeasy}.
\paragraph*{Proof Sketch~\ref{en:one}$\implies$\ref{en:decomp}.}
If $C$ is empty, then the decomposition is $\Dc = \seq{V}$ for a $(c_0,\emptyset)$-colourful graph where all paths satisfy the Rabin condition.
If $C$ is not empty, we first remove all vertices $A$ from $\Gc$ that can visit a vertex for which $c_0$ is a good colour. In the SCC decomposition of the graph induced by $V\setminus A$, each infinite path satisfies the Rabin condition, and therefore especially the infinite path which consists of all the vertices of some bottom SCC, $V_1$. Hence, there must be one colour $c$ that is not a bad colour for any vertex and is a good colour for at least some of the vertices $V_1$. One can therefore inductively construct a $(c,C\setminus\{c\})$-colourful decomposition $\Dc_1$ for the vertices of $V_1$.  Later, in the graph $\Gc$ without the vertices of $A$ and $V_1$ and all vertices $A_1$ from which all paths lead to $V_1$, we again get an other graph where all infinite paths satisfy the Rabin condition. This 
 graph, again by induction has a $(c_0,C)$-colourful Rabin decomposition $\Dc'$. We finally 
 `glue' together $\seq{A,(c,V_1,\Dc_1,A_1)}$ and  $\Dc'$ obtained above. 
%`glue' $A$, $\Dc_1$, $A_1$, and $\Dc'$ to obtain our final decomposition.
\paragraph*{Proof Sketch \ref{en:decomp}$\implies$\ref{en:measure}.} The proof follows a recursive construction of an $\labelset$-labelled $(c_0,C)$-colourful tree where the recursion is based on the structure of the decomposition. An example of how such a mapping to a tree is obtained from a picture is exemplified in \cref{fig:RGimage}. The decomposition $\Dc$ of the game $\Gc$ is $\tpl{A,(\rd,V_1,\Dc_1,A_1),(\bl,V_2,\Dc_2,A_2)}$. Some of the sets of the decomposition are indicated in \cref{fig:RGDecomp}. 
The measure obtained from the decomposition into the given tree is intuitive. For example, the measure obtained from the given decomposition of the game $\Gc$ is such that the vertex for which the colour $\blck$ is a good colour is mapped to the root of the tree. Similarly, this measure maps the vertex in $V_1$ for which the colour $\rd$ is a good colour  to the node $1\rd$ in the tree. The only vertex in $A_2\setminus V_2$ is mapped to the node $2\nocolournode$.  %recursively an $\Nb$-labelled $(c_0,C)$-colourful tree $\Tc$, into which vertices are mapped.
\paragraph*{Proof Sketch \ref{en:measure}$\implies$\ref{en:one}.}If there is a Rabin measure, each edge in the infinite path satisfies~\Rrr, as well as \Gone~or \Gtwo. For such an infinite path, we consider the infinite sequence of nodes of the colourful tree, 
obtained by taking the image of $\mu$ on the run. In this sequence obtained, consider the smallest node of the tree $t$ that is visited infinitely often, and let $c = \colouring(t)$. We show that $t$ is a common ancestor for all  elements of the sequence after a finite prefix. 
Since all edges satisfy \Gone~or \Gtwo, $c$ is a colour 
such that $c\in G_v$ for some $v$ visited infinitely often. As all edges satisfy \Rrr, we have $c\notin B_v$ for all vertices $v$ in the run after some finite prefix.
\begin{remark}
    A similar statement to the equivalence of item~\ref{en:one} and item ~\ref{en:decomp} has been proved in the work of Klarlund and Kozen~\cite{KK91}, however, a reader familiar with their work might have observed some differences in the definition of a measure as well as a colourful tree. Our definition of colourful trees is more restrictive than theirs. For instance, colourful trees in the work of Klarlund and Kozen have no restrictions about the colours along a path in a tree, i.e, in their definition, the trees can have the same colour along a path, and in fact only a partial colouring is required. However, an examination of their proof reveals that in the direction of the proof where they construct a Rabin measure, they inherently use a construction which produces a mapping into colourful trees as we have defined
    and therefore, it is enough to only consider such trees. We make this explicit and also prove Theorem~\ref{thm:measure}  in the appendix to suit our situation. 
\end{remark}

 \section{A lifting algorithm}~\label{sec:unilift}
 In this section, we define Rabin game formally and first show how such Rabin games also have a notion of a Rabin measure. 

Inspired by the breakthrough algorithms to solve parity games, Colcombet, Fijalkow, Gawrychowski, and Ohlmann~\cite{CFGO22} proposed a formalism for algorithms that solve games where one player has a positional strategy. They showed that if there is a special kind of graph homomorphism into a graph with a total order on its vertices, then one can obtain a lifting algorithm for such games. 
From their work \cite[Theorem~3.1]{CFGO22} combined with \cref{thm:mainRabin}, we can show that Rabin measures defined in our previous section can also be used  to provide a lifting algorithm for Rabin games. However, to make this work self-contained and to provide an explicit space-efficient algorithm using our non-trivial totally ordered set, we show how such lifting is performed step-by-step in this section. We believe our following section would help future implementation of such algorithms. %\thejaswiniside{Should I say this?}
% Furthermore, in the following 
% section, we pinpoint an additional 
% Combined with our results in \cref{sec:colourfulTrees}, 
% we can further minimally modify existing parity games solvers to directly implement one such algorithm. 

A $(c_0,C)$-colourful \emph{Rabin game} $\gamegraph$ consists of
%$= (\Gc,v_0,V = V_c\uplus V_e)$ consists of
~an \emph{arena} which is a $(c_0,C)$-colourful Rabin graph $\Gc$ with vertices $V$, 
%$= (V,E,\{B_v\}_{v\in V},\{G_v\}_{v\in V})$
a \emph{start vertex} $v_0\in V$, and 
a partition of $V$ into $V_c$ and $ V_e$, the vertices  of two players, whom we call \Controller~and \Environment, respectively. 

A  \emph{positional strategy} $\sigma$ for \Controller~over the game graph is a subset of edges outgoing from each of \Controller's set of vertices $V_c$. We denote the graph restricted to a strategy $\sigma$ for the \Controller~by $\Gc|_\sigma$ and it is defined as the Rabin graph over the same vertex set with a new edge relation which contains exactly the edges in $\sigma$ along with all the edges from all vertices belonging to \Environment.

The Rabin game $\gamegraph$ is \emph{winning} for the \Controller~if and only if there exists a positional strategy $\sigma$ for the \Controller~where, all infinite paths starting from $v_0$ in $\Gc|_\sigma$ satisfy the Rabin condition. 
%Having characterised Rabin graphs where all paths satisfy the Rabin condition, we proceed to showing how this helps us give an 
We describe an algorithm that identifies whether a Rabin game $\gamegraph$ is winning for \Controller,
using Rabin measures on graphs. 

\begin{remark}
We only consider strategies of the \Controller~that are positional, but this is enough from the results of Emerson and Jutla~\cite{EJ99}, which shows that the \Controller~always has a positional winning strategy in Rabin games if there is any winning strategy at all.
\end{remark}

\paragraph*{Consistency in games.}
Consider a $(c_0,C)$-colourful Rabin game $\gamegraph$. 
Let $\mu$ be a function from $V$, the vertices of the game graph to an $\labelset$-labelled $(c_0,C)$-colourful tree $\Lc$. %We simply extend this to games.
%We had defined what it meant for an edge in a Rabin graph to be consistent with respect to a map in Section~\ref{sec:shapeofRG}. 
We simply extend the definition of consistency from graphs to games by defining 
a vertex to be \emph{consistent} with respect to $\mu$ in $\gamegraph$ if either it belongs to the \Environment~and all outgoing edges from it are consistent in $\Gc$ or if it belongs to the~\Controller~and there is at least one outgoing edge that is consistent in $\Gc$.
A  map $\mu$ from $V$ to a $\Lc^\top$ 
is a \emph{Rabin measure} for a $(c_0,C)$-colourful Rabin game $\gamegraph$ if and only if \emph{all} vertices are consistent with respect to $\mu$.

\paragraph{An overview of the algorithm.}We describe an algorithm that identifies whether a Rabin game $\gamegraph$ is winning for  \Controller,
using Rabin measures defined earlier for Rabin graphs. The basic principle in the algorithm is that given a colourful tree, the algorithm finds if there is a Rabin measure that maps vertices of the game into nodes of that tree. The algorithm does so by starting with the smallest map (all vertices are mapped to the root of this tree) and then at each step, if a vertex is not consistent, increase the value of this map just at this vertex which is not consistent. The value is modified (increased) until either all vertices are consistent, or the value cannot be increased anymore. 

Toward our goal of formally defining this algorithm, we define monotonic, inflationary operators on the set of all maps from vertices of a game to a tree such that the simultaneous fixpoints of these operators exactly correspond to a Rabin measure.

% The above lemma shows that there is a unique smallest Rabin measure, which is given by taking the minimum of all Rabin measures. 
%
%We define monotonic inflationary operators. 
Consider  a Rabin measure $\mu$ which is a function mapping the vertices $V$ of  a $(c_0,C)$-colourful Rabin game $\gamegraph$ into an $\labelset$-labelled $(c_0,C)$-colourful tree $\Lc$. 
We define a function $\lift_\mu$, which maps edges $E$ of the arena of the game to $\Lc^\top$. For an edge $u\rightarrow v$ of $\Gc$, we define $\lift_\mu(u,v)$ to be the smallest element $t$ in $\Lc^\top$ 
such that (1) $t\succeq \mu(u)$ 
and (2) edge $u\rightarrow v$ is consistent with respect to the mapping $\mu[u:=t]$,
where we use the notation $\mu[u:=t]$ to indicate the mapping $\mu'$ where $\mu'(x) = \mu(x)$ if $x\neq u$ and $\mu'(x) = t$ if $x=u$.

 For each vertex $v$, we define an operator $\biglift_v$ on the lattice of all maps from $V$ to $\Lc^\top$.
 The operator $\biglift_v$ only modifies an input map $\mu$ at $v$ and nowhere else. We define
    $$\biglift_v(\mu)(u) = \begin{cases}
        \mu(u)  & \text{for }u\neq v\\
        \min_{(v,w)\in E}\left\{\lift_\mu(v,w)\right\} & \text{if } u=v\in V_c\\
        \max_{(v,w)\in E}\left\{\lift_\mu(v,w)\right\} & \text{if } u=v\in V_e
\end{cases}$$

\begin{restatable}{proposition}{minoftwo}~\label{lemma:minoftwo}
The function $\biglift_v$ is monotonic for each $v$. 
\end{restatable}
The above proposition follows from our definition of the $\biglift_v$ function. Now that we know that each $\biglift_v$ is inflationary and monotonic. 
Therefore, the simultaneous least fixpoint of $\biglift_v$ on the map $\mu$, which maps all vertices to the root of $\Lc$ exists (from the Knaster-Tarski theorem~\cite{Tar55}). 
We can moreover state the following proposition that such fixpoints correspond to the Rabin measures, which almost follows from our definitions. 
\begin{restatable}{proposition}{fpLiftvisRM}\label{lemma:fpLiftvisRM}
    For a $(c_0,C)$-colourful Rabin game $\gamegraph$ where the vertex set is $V$ and a fixed $\labelset$-labelled $(c_0,C)$-colourful tree $\Lc$, 
    \begin{itemize}
        \item any simultaneous fixpoint of the set of functions $\biglift_v$ for all $v\in V$ is a Rabin measure;
        \item any Rabin measure is a simultaneous fixpoint of~ $\biglift_v$ for all $v\in V$.
    \end{itemize} 
\end{restatable}

Our algorithm, like any other progress-measure algorithm, computes a fixpoint and is described as follows. 
\begin{algorithm}
\caption{The lifting algorithm on game $(c_0,C)$-colourful Rabin game $\gamegraph$ with vertices $V$ to tree $\Lc$}\label{algo:RabinLifting}
\begin{algorithmic}[1]
\Require For each $v\in V$,  $\mu(v)$ is declared to be root in $\Lc$
\While{there is some vertex $v$ that is inconsistent with respect to $\mu$.}
    \State $\mu\gets \biglift_{v}(\mu)$.
\EndWhile 
\State \Return $\mu$
\end{algorithmic}
\end{algorithm}
The correctness follows from \cref{lemma:minoftwo,lemma:fpLiftvisRM}.

\begin{remark}~\label{remark:rabinmeasureembed}
    If there is a $(c_0, C)$-colourful Rabin game $\gamegraph$ %, a $(c_0,C)$-colourful graph 
    and an $\labelset$-labelled $(c_0,C)$-colourful tree $\Lc'$, such that there is a Rabin measure $\mu'$ from $V$ to $\Lc'$, 
    and $\Lc$ embeds $\Lc'$, then there is also a Rabin measure $\mu$ to $\Lc$ such that all the elements that are not mapped to $\top$ by $\mu'$ are still not mapped to $\top$ by $\mu$. This map is obtained by composing $\mu'$ with the embedding of $\Lc'$ into $\Lc$. 
\end{remark}

\paragraph*{Runtime.}
For a finer analysis of the runtime, we need to understand the size of the lattice where the lifting algorithm takes place. In this section however, we restrict ourselves to analysing the runtime of our algorithm for a fixed $\Lc$. We write $|\Lc|$ to represent the number of nodes in the labelled tree $\Lc$. 
We write $n$ to denote the number of vertices in a Rabin game, $m$ to denote the number of edges, and $k = |C\cup \set{c_0}|$ to denote the 
number of colours. %In the following lemma, we  assume each node of the tree $\Lc$ is represented as an element of the a sequence in $(\labelset\times C)^*$.

\begin{restatable}{lemma}{bigLift}~\label{lemma:bigLift}
Given a mapping from the vertices of a $(c_0,C)$-colourful Rabin game $\gamegraph$ to an $\labelset$-labelled $(c_0,C)$-colourful tree $\Lc$, the value of ~$\biglift_v(\mu)(v)$ can be computed in time $O\left(\deg(v)\cdot T_{\mynext}\right)$, where 
$\deg(v)$ is the degree (number of outgoing edges) of $v$ and $T_{\mynext}$ is defined as the maximum of the time taken to
    \begin{itemize}
        \item  make a linear pass on a node in $\Lc$ (assuming the node is represented by a sequence of elements of~ $\labelset\times C$),
        \item  compute the next node in $\Lc$, and
        \item  find the next node that uses colours only from $C'\cup \{\bot\}$ for a given node $t\in \Lc$ and subset of colours $C'\subseteq C$ such that  $\colouring(t) \in C'$. 
    \end{itemize}
\end{restatable}
The proof of the above lemma reduces to arguing carefully by analysing cases, that using these above subroutines, we can find the node larger than $\biglift_v(\mu)(v)$ in the given tree that satisfies the conditions \Rrr~along with at least one of \Gone~or~\Gtwo. %We argue in the proof of the above lemma, given in appendix, that we can find the next node that does not use a bad colour of $v$ (using item 3 of \cref{lemma:bigLift}) and find its first child larger than the current node value of $\mu(v)$ that either satisfies \Gone~or~\Gtwo using the operations described above.     
%\end{proof}

First, we observe that performing $\biglift_v$ on the mapping strictly increases the mapping for a vertex that is not consistent. Each operation of $\biglift_v$ also calls at most $\deg(v)$ many calls of $\lift_{\mu}(v,u)$ for some edge $v\rightarrow u$.
%, and the time taken for each lift is proportional to the increase in the labelling of the vertex that it is mapped to. 
Suppose each operation $\lift_{\mu}(v)$ takes time $T_{\mynext}$, to find the value of~ $\biglift_v(\mu)(v)$ takes time at most $\deg(v)\cdot T_{\mynext}$. Since each non-trivial application of $\biglift_v$ strictly increases the value that $v$ is mapped to, it can be called at most as many times as the number of nodes in tree $\Lc$, this ensures that the time taken 
is 
$$\sum_{v\in V}\mathrm{deg}(v)|\Lc|\tpl{T_{\mynext}} \in O\left(m|\Lc|T_{\mynext}\right)$$
 where $m$ denotes the number of edges. 
We finally conclude this section with the following theorem, which follows from \cref{lemma:bigLift}. 
\begin{restatable}{theorem}{liftingalgo}~\label{thm:liftingalgo}
   For a $(c_0, C)$-colourful Rabin game $\gamegraph$ with $n$ vertices and $m$ edges, and an $\labelset$-labelled $(c_0,C)$-colourful tree $\Lc$, Algorithm~\ref{algo:RabinLifting} on $(\gamegraph,\Lc)$ returns the smallest Rabin measure to $\Lc^\top$ in time $O\left(m|\Lc|T_{\mynext}\right)$ where $T_{\mynext}$ 
   is as defined in \cref{lemma:bigLift}
   and $|\Lc|$ denotes the number of nodes in $\Lc$.
\end{restatable}

%Plugging in the universal trees obtained from Section~\ref{sec:colourfulTrees}, this gives a $O\left(\log n \log kmn^2k^22^kk!\right)$, which is $\Tilde{O}\left(mn^3k2^kk!\right)$.

 \section{Small Colourful Universal Trees}\label{sec:colourfulTrees}

In the previous section, we concluded
 that 
our algorithm identifies correctly the smallest Rabin measure into a fixed labelled colourful tree $\Lc$.
However, from Theorem~\ref{thm:measure}, \emph{there exists} a Rabin measure into an $\labelset$-labelled $(c_0,C)$-colourful tree $\Lc$ with at most $n$ leaves. Observe that we only need to consider $n$ leaves of $\Lc$ which correspond exactly to the image of the Rabin measure. 
Therefore, for a Rabin game, there is a Rabin measure into $\Lc^\top$ where all start vertices from which the game is winning for \Controller~are not mapped to $\top$. 
In order for the algorithm to successfully determine the winner of all $(c_0,C)$-colourful Rabin games with $n$ vertices, we need to ensure that the tree $\Lc$ used in Algorithm~\ref{algo:RabinLifting} 
would be able to embed all $(c_0,C)$-colourful trees with $n$ leaves. Since the runtime is linearly dependent on the tree size, smaller trees that satisfy the above property are desirable. 

%that our algorithm identifies the smallest Rabin measure into a labelled colourful tree. 
%Indeed, this expects us to know a colourful tree $\Lc$ beforehand that would ensure that the smallest Rabin measure into such an $\Lc$ does not map any vertex winning for the controller to $\top$. 
%But we also remarked that it is enough to know \emph{some} tree that embeds such an $\Lc$.

We now show that we can obtain colourful \emph{universal} trees, 
i.e., colourful trees that are large enough to embed \emph{any} $(c_0,C)$-colourful $\Lc$ with $n$-nodes.
We also modify the technique of succinct universal trees of Jurdzi\'nski and Lazi\'c~\cite{JL17} to encode each node of these colourful universal trees using polynomial space, which helps navigate these labelled colourful trees efficiently.

%\begin{definition}[$n$-Universal $C$-colourful tree]
%A $C$-colourful tree $\Uc$ is said to be $n$-universal, if it can embed any partial $C$-colourful tree $\Tc$ with at most $n$ leaves.
%\end{definition}
%Let $[h] = \{1,\dots,h\}$, if $h\geq 1$, and $[0] = \emptyset$.

%We extend the definition  to $n$-Universal partial $C$-colourful trees.
\paragraph*{Colourful universal trees.}
A $(c_0,C)$-colourful tree $\Uc$ is \emph{$n$-universal}, if it embeds \emph{any} $(c_0,C)$-colourful tree $\Tc$ with at most $n$ leaves. 
We henceforth assume that the set $C$ consists exactly of the colours $c_1,\dots,c_{h}$, with the fixed ordering $c_1<c_2<\dots<c_{h}$ on the colours, and use $k$ to denote $h+1$.

A n\"{a}ive attempt at constructing an $n$-universal $(c_0,C)$-colourful tree could be to take all possible  $(c_0,C)$-colourful trees with at most $n$ leaves with the root colour $c_0$ and concatenate them. 
Clearly, such an $n$-universal $(c_0,C)$-colourful tree can be created as there are only finitely many such trees up to isomorphism (for a fixed $C$ and $n$). 
But of course, this tree is not only large, but can also be difficult to navigate. 
A more tractable attempt is to construct a tree that branches $n \cdot h$ many times at the root. 
The subtrees at the root that occur from this $n \cdot h$ branching have $n$ repetitions of the $h$ colours $c_1,c_2,\dots,c_{h}$, in that order. 
Each of the children in-turn branch into $n\cdot (h-1)$ many times similarly, 
thus creating a tree of size bounded by $n^hh!$. 
We claim that indeed such a tree was exactly the one underlying the algorithm of Piterman and Pnueli~\cite{PP06}, 
which led to their $O(mn^{k+1}kk!)$ algorithm.
%, with $k$ denoting a $|C\cup \{c_0\}|$ whose runtime is within a polynomial factor of the size of the tree described above.\thejaswini{Is this better?} 

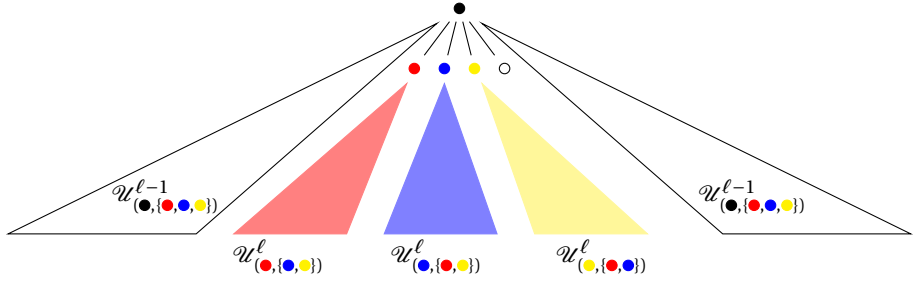
\begin{figure}[ht]
    \centering
    \begin{tikzpicture}[nodes={draw, draw=white!0=02}]
\node at (0,0) {$\blck$}
        child[level distance=8mm, sibling distance =4mm] { node {$\rd$} 
        }
        child[level distance=8mm, sibling distance = 4mm] { node{$\bl$}
        }
        child[level distance=8mm, sibling distance = 4mm] { node{$\gr$}
        }
        child[level distance=8mm, sibling distance = 4mm] { node{$\nocolournode$}
        };

        \coordinate (l0) at (-0.3,-0.2);
        \coordinate (l1) at (-6,-3);
        \coordinate (l2) at (-3.5,-3);
        
        \coordinate (r0) at (0.3,-0.2);
        \coordinate (r1) at (6,-3);
        \coordinate (r2) at (3.5,-3);
        
        \coordinate (a0) at (-0.2,-1);
        \coordinate (a1) at (-1,-3);
        \coordinate (a2) at (0.5,-3);

        \coordinate (b0) at (-0.7,-1);
        \coordinate (b1) at (-3,-3);
        \coordinate (b2) at (-1.5,-3);
        
        \coordinate (c0) at (0.3,-1);
        \coordinate (c1) at (2.5,-3);
        \coordinate (c2) at (1,-3);
        \filldraw[draw=red!50, fill=red!50] (b0) -- (b1) -- (b2) -- cycle;
        \filldraw[draw=blue!50, fill=blue!50] (a0) -- (a1) -- (a2) -- cycle;
        \filldraw[draw=yellow!50, fill=yellow!50] (c0) -- (c1) -- (c2) -- cycle;
        
        \filldraw[draw=black!100, fill=red!0] (l0) -- (l1) -- (l2) -- cycle;
        \filldraw[draw=black!100, fill=red!0] (r0) -- (r1) -- (r2) -- cycle;
        %\filldraw[draw=black, fill=red!20, line width=1.5pt,densely dotted]  (r0) to[out=300, in=130] (si) -- (s1) -- cycle; 
        \node[draw=none] (lefte1) at (-3.9,-2.5) {$\Uc^{\ell-1}_{(\blck,\{\rd,\bl,\gr\})}$}; 
        \node[draw=none] (righte1) at (3.9,-2.5) {$\Uc^{\ell-1}_{(\blck,\{\rd,\bl,\gr\})}$}; 
        \node[draw=none] (redtree) at (-2.4,-3.3) {$\Uc^{\ell}_{(\rd,\{\bl,\gr\})}$};
        \node[draw=none] (bluetree) at (-0.3,-3.3) {$\Uc^{\ell}_{(\bl,\{\rd,\gr\})}$};
        \node[draw=none] (yellowtree) at (1.9,-3.3) {$\Uc^{\ell}_{(\gr,\{\rd,\bl\})}$};
\end{tikzpicture}    
        \caption{Inductive construction of a smaller colourful $n$-universal tree}
    \label{fig:UniversalColourful}
\end{figure}
Below, we give a more involved construction of a significantly smaller universal tree. 
In our construction, we inductively describe such a $(c_0,C)$-colourful $n$-universal tree, which we call $\Uc^\ell_C$, for a fixed $n \leq 2^\ell$.
\begin{itemize}
    \item if $C = \emptyset$, then there is exactly one tree to embed, and therefore $$\Uc^\ell_{(c_0,C)} = \tpl{c_0,\seq{\left(\bot,\seq{}\right)^{2^\ell}}}$$
    \item if $\ell=0$, then the tree to be embedded has exactly one leaf and therefore, for each colour $c_i$ in $C$, we have a child of colour $c_i$ which hosts a subtree whose colour at the root is $c_i$. This is defined inductively as
    $$\Uc^0_{(c_0,C)} \:= \tpl{c_0,\seq{\Uc_{(c_1,C_1)}^0,\dots,\Uc_{(c_h,C_h)}^0,\left(\bot,\seq{}\right)}}$$ where   $C_i$ is $C\setminus \{c_i\}$. 
    \item if $C \neq \emptyset$ and $\ell > 0$, then we define the coloured tree to be two copies of an $n/2$-universal tree, and $h$ many copies of the $n$-universal tree where one colour is dropped each time. More formally, 
    $$\Uc^\ell_{(c_0,C)} \:=
    \Uc_{(c_0,C)}^{\ell-1}\cdot
\tpl{c_0,\seq{\Uc_{(c_1,C_1)}^\ell,\dots,\Uc_{(c_h,C_h)}^\ell,\left(\bot,\seq{}\right)}}
\cdot\Uc_{(c_0,C)}^{\ell-1}.$$
\end{itemize}
In \cref{fig:UniversalColourful}, we demonstrate how the inductive construction is done if $c_0 = \blck$ and the set of colours is $C=\{\rd,\bl,\gr\}$. To the left and right are the $(\blck, C)$-colorful $n/2$-universal trees and between them, there are $|C|$ many $n$-universal trees each of which uses one fewer colour and one node with just the dummy colour represented there by $\nocolournode$. 

\begin{restatable}{theorem}{universal}\label{thm:universal}
For $C\neq \emptyset$, and $k= |C|+1$, $\Uc^\ell_{(c_0,C)}$ constructed is a $(c_0,C)$-colourful $n$-universal tree with at most $$nk! \left(\min\left\{n2^k , {\ell +k\choose k-1}\right\}\right)$$ many leaves, where $\ell = \lceil\log n\rceil$.
\end{restatable}
\begin{proof}
Firstly, we show that $\Uc^\ell_{(c_0,C)}$  is $(c_0,C)$-colourful $n$-universal tree (in Proposition~\ref{prop:colourfuluniversal}). Then we prove using induction that $\Uc_{(c_0,C)}^\ell$ 
%To prove the above theorem, we instead prove 
%our trees
has at most $2^k\cdot k!\cdot 4^\ell$ leaves in
Lemma~\ref{lemma:smallcolourful} and at most  ${\ell +k\choose k-1}\cdot2^\ell\cdot k!$ leaves in Lemma~\ref{lemma:smallcolourfulcombibound} leading to the proof of our theorem. 
\end{proof}

% \subsection{Lower bounds on the size of colourful universal trees}
In fact, we have a lower bound for $n$-universal $(c_0,C)$-colourful trees, which is within a polynomial factor of the upper bound.

It is known from the work of Calude et al., as well as from Casares et al.~\cite{CJKLS22,CPPST23} that there are no algorithms that solve Rabin games in time $n^{O(1)}\cdot 2^{o(k\log k)}$. But observe that this does not exclude algorithms which is dependant on $k!$ by only a constant smaller than 1 in the exponent. 
We have improved the current state-of-the art from $2+o(1)$ to $1+o(1)$ in the exponent. A natural question to ask would be if the $k!$ component can be reduced further. We show below that we cannot improve our running time much further using our techniques. 
\begin{lemma}[Lower bound]~\label{lemma:lowerboundcolourfultree}
Any $n$-universal $(c_0,C)$-colourful tree must have size at least ${\ell + k - 1 \choose{\ell}}(k-1)!$ where $k = |C|+1$ and $\ell = \lfloor\log n\rfloor$.
\end{lemma}
\begin{proof}\label{prf:lowerboundcolourfultree}
Fix a permutation $c_{i_1},\dots,c_{i_h}$ of the colours in $C$ and consider any tree with $n$ leaves where the order of colours from the root to the leaf is exactly the same as the given permutation. Moreover, we assume that the leaves all have the same depth from the root. This tree must have size at least the size of a $2^\ell$-universal tree (defined for ordered tree without colours). Such universal trees have size at least ${\ell + h - 1 \choose{h-1}}$ in the work of Czerwi\'{n}ski et. al~\cite{CDFJLP19}. For each choice of permutation, the universal tree restricted to that permutation must have size ${\ell + h - 1 \choose{h-1}}$. Furthermore, two universal trees obtained by fixing different permutations cannot share a leaf since distinct colours are assigned to some ancestor of such leafs. Therefore, we obtain a lower bound of ${\ell + h - 1 \choose{h-1}}h!$ on the size of any $(c_0,C)$-colourful $n$-universal trees. 
  
This immediately gives us the bound ${\ell + k - 2 \choose{\ell}}(k-1)!$ for $k = |C|+1$. Our lower bound also matches one of the upper bounds of our construction up to a polynomial factor in $n$ and $k$.
\end{proof}

\paragraph*{Labelling Colourful Universal Trees.}
Here, we give a labelling of a universal colourful tree described in the previous section by giving an $\words$-labelling of any $(c_0,C)$-colourful tree where the set $\words = \{0,1\}^*$. We let $\varepsilon$ denote the empty string in $ \{0,1\}^*$.  We define the ordering on $\{0,1\}^*$ as follows, similar to the succinct encoding of ordered trees~\cite{JL17}: $0<\varepsilon<1$ and for $b_1,b_2\in \{0,1\}$ we have $b_1\cdot w_1< b_2\cdot w_2$ if and only if $b_1<b_2$ or $b_1 = b_2$ and $w_1<w_2$.

Any node $t$ in a $\words$-labelled $(c_0,C)$-colourful tree can be represented by a word generated by the following regular expression $$\{0,1\}^*c_{i_1}\cdot\{0,1\}^*c_{i_2}\cdot\ldots\cdot\{0,1\}^*c_{i_m}$$ where $c_{i_j}\neq c_{i_k}$ if $j\neq k$ and $c_{i_j} = \bot$ if and only if $j = m$.
We call the number of $0$s and $1$s occurring in the word, \emph{the number of bits used to label $t$}.
We show in the following lemma that it is possible to have a labelling of our colourful universal tree $\Uc_{(c_0,C)}^\ell$ such that the labelling of each node in it is `short'.% tree such that for all nodes $t$ in the tree, the number of bits used to label $t$ is at most $\ell.

\begin{restatable}{lemma}{succinctcolourfullabeling}~\label{lemma:succinctcolourfullabeling}
There is a $\words$-labelling of the tree $\Uc_{(c_0,C)}^\ell$, denoted by  $\Lc_{C}^\ell$ such that the number of bits used to label any node of   $\Lc_{C}^\ell$  is at most $\ell$.
\end{restatable}
\begin{proof}[sketch]For 
$\ell>0$, we have $\Uc^{\ell}_{(c_0,C)} = \Uc^{\ell -1}_{(c_0,C)}\cdot\tpl{c_0,\seq{\Uc_{(c_1,C_1)}^\ell, \dots, \Uc_{(c_h,C_h)}}} \cdot \Uc^{\ell -1}_{(c_0,C)}$. We obtain recursively a labelling of $\Uc^{\ell -1}_{(c_0,C)}$ and append the bit 0 for the copy on the left and append with 1 for the copy on the right. For all the labellings of $\Uc_{(c_i,C_i)}^\ell$, we add the element $\varepsilon\cdot c_i$ as a prefix.
\end{proof}
%Finally, (maybe?) remark that in fact all labellings of the form $(\omega_1\cdot c_{i_1},\dots, \omega_m\cdot c_{i_m})$ which uses at most $\ell$ bits, each where $c_{i_j}$ is different from other colours and $c_{i_j} = \bot$ if and only if $j = m$ are nodes in this labelled tree.

% Consider the set of all labellings of the form $(\omega_1\cdot c_{i_1},\dots, \omega_m\cdot c_{i_m})$ and $\sum_i|\omega_i| \leq \ell$ where $c_{i_j}\neq c_{i_k}$, where each $c_{i_j}\in C\cup\{\bot\}$, $c_{i_j}\neq c_{i_k}$ if $i\neq j$, and $c_{i_j} = \bot$ if and only if $j = m$.
% Moreover we restrict it to nodes such that $\sum_i |\omega_i| \leq \ell$. This forms an $\words$-labelling of a $(c_0,C)$-colourful $2^\ell$-universal tree and exactly corresponds to a labelling of the colourful tree $\Uc^{\ell}_{(c_0,C)}$.% and therefore be able to embed any  $\Tc$ with at most $2^\ell$ leaves.
We rigorously prove this in the appendix, but only state here that the three operations  defined in the statement of \cref{lemma:bigLift} can be computed in time $O(k\ell\log k)$ (denoted by $T_\mynext$), where $k = |C|+1$ (Lemma~\ref{lemma:nextnodeOne} and Proposition~\ref{prop:nextnodeTwo}).
\begin{restatable}{theorem}{mainRabin}
~\label{thm:mainRabin}
    Finding the winner in a $(c_0,C)$-colourful Rabin game with $n$ vertices, $m$ edges, and $k = |C|+1$, takes time $$\widetilde{O}\left(mnk^2k!\min\left\{n2^k,{\lceil \log n\rceil+k\choose k-1}\right\}\right)$$ and $O(nk\log k \log n)$ space.
\end{restatable}
\begin{proof}
    We know that the lifting Algorithm~\ref{algo:RabinLifting} for a $(c_0,C)$-colourful tree finds the Rabin measure into the tree $\Lc$ in time $O(m|\Lc|T_\mynext)$ from \cref{thm:liftingalgo}. %Let $\ell = \lceil\log n\rceil$.
    %If some vertex $v$ is winning for the \Controller, then there must be a Rabin measure into
    For a game with $n$ vertices, we instantiate the algorithm with $\Lc$ being the $\words$-labelling of the $(c_0,C)$-colourful $2^\ell$-universal tree $\Uc_{(c_0,C)}^{\ell}$ constructed, where $\ell =\lceil\log(n)\rceil$. The tree
    $\Lc$ therefore has at most $\left(nk!\min\left\{n2^k,{\lceil \log n\rceil+k\choose k-1}\right\}\right)$ many leaves from Theorem~\ref{thm:universal}, and hence at most $k$ times as many nodes. 
    Moreover, the time taken to navigate the tree $T_\mynext$ is at most $O(k\ell\log k)$.
    The space required by the algorithm at each step is just the space required to store the map. This takes $O(k\log k \log n)$ for each of the $n$ vertices, giving us the desired space complexity.
\end{proof}

% \section{Fair Rabin Games}~\label{sec:fairness}
 %\input{5Fairness.tex}
% %
 \section{Conclusions, Discussion and Future Work}
 We have shown an algorithm for Rabin games that requires almost quadratic space and takes time that is polynomial in $n$ and $(k!)^{1+o(1)}$. 
Significantly more asymptotic improvement to the running time may be difficult, as it was shown in the work of Calude et al.~\cite{CJKLS22,CPPST23} that there are no algorithms to solve Rabin games (as well as Muller games) in time $n^{O(1)}\cdot 2^{o(k\log k)}$ unless the Exponential Time Hypothesis fails (informally, it is the assumption that 3-SAT has no sub-exponential algorithms). However, improvements in the exponents of the parameter $k!$, which contributes to the majority of the running time would prove useful in any algorithm that solves Rabin games. We have shown that using colourful universal trees cannot provide a significant improvement bound because of the 
$k!^{1+o(1)}$ lowerbound on the size of such a tree. However, any technique that improves, even on a few targeted cases, this $1+o(1)$ bound could lead to faster algorithms. For instance, the recent unpublished work of Liang, Khoussainov, and Xiao~\cite{LKX23} improve the running time for specific values of $k$, where the size of $k$ is large (comparable to $n$).

While we focus on the theoretical advance in this paper, an obvious future direction is to implement the algorithm. 
There are tools that convert LTL specifications to Rabin automata---such as Rabinizer 4 \cite{KMSZ18}. 
It will be interesting to see if solving the obtained Rabin games using our algorithms outperforms converting them instead to parity games and then using state-of-the-art parity game solvers such as Oink~\cite{vDij18} framework. We believe improvement in state space of solving Rabin games through our paper might lead to more efficient algorithms for the problem of reactive synthesis of LTL formulas. 
% Indeed, in recent work by Banerjee et al. on Rabin game solving in the context of controller synthesis~\cite{BMMSS22}, the authors show experimental results with $k=3$  but their tool does not finish for instances with $k=4$. For such instances, we believe implementing our algorithm might be fruitful endeavour, as the improvement of the running time of the algorithm from $n^2mk!^{2+o(1)}$ to $n^2mk!^{1+o(1)}$ (instead of $4!^2$ the run time will be $4! = 24$), already gives us an improvement of 24 times. This is the difference between an instance requiring 24 hours vs 1 hour. As $k$ increases, our algorithm’s improvement is highlighted even further. For even larger cases of $k = 8$, the value of $k!$ is in the order of $10^5$. The best-known previous algorithm would take time proportional to an order of magnitude of $10^10$ or $10^11$ ($k!^{2+o(1)}$). Whereas, our algorithm, which takes time in the order of magnitude of $10^5$ might make such cases tractable. 
% Indeed, the new quasipolynomial algorithms tend to perform worse in practice than cleverly implemented 
% recursive exponential algorithms and we need more proof to verify if this trend is true for Rabin games. 

Our algorithm, like other progress measure algorithms, can display worst-case behaviour in certain asymmetric examples. To show a vertex is losing for~\Controller, the measure needs to increase until it reaches $\top$. This lack of symmetric treatment of the players by our algorithm might lead to worst case behaviour on several examples. But circumventing this problem by constructing similar measures for~\Environment~in the hopes of finding a symmetric algorithm is not as straightforward, as \Environment~does not have a positional strategy in this game. 
%Such a measure would also have to encode this memory. 

In a different direction, symbolic algorithms for parity games are either implicitly or explicitly guided by universal trees~\cite{CDHS18,JMT22} constructed for both players. 
We believe with some effort, our small colourful universal trees can be exploited to make symbolic algorithms to solve Rabin 
games. 
One such algorithm would look like an asymmetric variation of the universal algorithm in the work of Jurdzi\'nski, Morvan, and Thejaswini~\cite{JMT22} for parity games, combined with our construction of colourful universal trees. 
Indeed, we already have a definition of colourful decompositions which one might hope to obtain as an end-result of such a recursive symbolic algorithm. 

% One can also analyse parameters with respect to which Rabin games are fixed parameter tractable. Some examples of such parameters include, graphs of bounded tree-width, entanglement, DAG-width and Kelly-width.
% If such algorithms are found, we can potentially benefit both theoretically and practically with some clever pre-processing steps, combined with such algorithms. Modifying the input game graph or the acceptance condition to significantly reduce a suitable parameter, or even find restriction of trees required to perform the lifting on, could lead to an improvement in the performance of the algorithm. 

% \begin{textblock}{20}(0, 12.7)
% \includegraphics[width=35px]{logo-erc}%
% \end{textblock}
% \begin{textblock}{20}(0, 13.6)
% \includegraphics[width=35px]{logo-eu}%
% \end{textblock}
\paragraph*{Acknowledgements.}
We would like to thank Marcin Jurdzi\'{n}ski and Anne-Kathrin Schmuck for valuable discussions and references. We also thank Aditya Prakash for his valuable comments and, in particular, for reading the section on colourful trees despite his colour blindness.
\bibliographystyle{splncs04}
\bibliography{biblio}

\newpage
 \appendix
 \section{Appendix for Section~\ref{sec:shapeofRG}}
  
%%%%%%%%%%%%%-------------------------------------
%%%%%%%%%%%%%%---------------------------------

\begin{restatable}{lemma}{winningdecomposition}
\label{lemma:winningdecomposition}
Let $\Gc$ be a $(c_0,C)$-colourful Rabin graph where all infinite paths satisfy the Rabin condition, then there is a $(c_0,C)$-colourful decomposition of $\Gc$.
\end{restatable}

\begin{proof}[Proof of Lemma~\ref{lemma:winningdecomposition}]\label{prf:winningdecomposition}
  We construct such a decomposition, by inducting on $|C|$ and the number of vertices in $\Gc$.
\paragraph*{Base Case}   If $C=\emptyset$, then for all vertices $v$,  $B_v = \emptyset$. %Indeed, if for any vertex in the SCC, if $B_v = C$, there is a path that visits $v$ infinitely often. 
Observe also that all paths in the SCCs after finitely many steps must visit a vertex $v$ such that $G_v = \{c_0\}$. This is because all paths in $\Gc$ satisfy the Rabin condition, there is no infinite path such that $G_v = \emptyset$ for all $v$ along the path. 
The $(c_0,C)$-colourful decomposition is just $\seq{V}$.

\paragraph*{Induction Hypothesis}
For all $(c_0,C)$-colourful Rabin graphs $\Gc$ where all infinite paths satisfy the Rabin condition and $|C|<k$ or $|V|<n$.  Then there is a $(c_0,C)$-colourful decomposition 
$$\Dc\:= \seq{A, \tpl{c_{1},V_1, \Dc_1,A_1}, \dots, \tpl{c_{j},V_j,\Dc_j,A_j}}$$
where
 $c_0\notin B_v$ for all $v\in V$, and for all $v\in V$, if $c_0\in G_v$ then $v\in A$.
\paragraph*{Induction Step}
Suppose $|C| = k$, and the induction hypothesis holds. 
Consider all vertices $R \:= \{v\mid c_0\in G_v\}$, and let $A\supseteq R$ be the set of vertices from which every infinite path starting from $A$ visits some vertex from $R$. 
The subgraph induced by $W_1 = V\setminus A$, which is a subgraph of $\Gc$ also satisfies the property that all vertices have an outgoing edge. More importantly, all infinite paths in it satisfy the Rabin condition. Moreover, there are no vertices $v$ such that $c_0\in G_v$ or $c_0\in B_v$ for $v\in W_1$.

Consider an SCC decomposition of the graph induced by $V\setminus A$.  Consider a bottom SCC (an SCC from which there is no path to other SCCs) $V_1$ of the graph induced by $V\setminus A$.
Consider a path $\pi$ such that the set of all vertices visited by $\pi$ infinitely often is exactly $V_1$. This path satisfies the Rabin condition, which implies there is some colour $c_1$ such that $c_1\notin B_v$ for all $v\in V_1$ and $c_1\in G_v$ for some $v\in V_1$. 

Therefore, by induction, there is a $(c_1,C\setminus\{c_1\})$-colourful decomposition of $V_1$, say $\Dc_1$.
Let $A_1$ denote all the vertices in $V \setminus A$ from which all infinite paths lead to a vertex in $V_1$. 

Now consider the game $W_2 \:= (V\setminus A)\setminus A_1$, which has fewer vertices. We know again that $W_2$ is a subgraph of $\Gc$ and therefore there are no vertices $v$ such that $c_0\in G_v$ or $c_0\in B_v$ for $v\in W_2$, there must be a $(c_0,C)$-colourful decomposition.

Let this decomposition be: $$\Dc'\:= \seq{\emptyset, \tpl{c_{2},V_2, \Dc_2,A_2}, \dots, \tpl{c_{j},V_j,\Dc_j,A_j}}$$
Note that the top set of vertices is $\emptyset$, by induction hypothesis since there are no vertices $v$ where $c_0$ is a good colour or a bad colour of $v$.

We claim that $$\Dc\:= \seq{A, \tpl{c_{1},V_1, \Dc_1,A_1},  \tpl{c_{2},V_2, \Dc_2,A_2}, \dots, \tpl{c_{j},V_j,\Dc_j,A_j}}$$ thus constructed from the sets defined above is a $(c_0,C)$-colourful decomposition.

It is routine to verify that it satisfies all the properties of a decomposition by construction. 
\end{proof}

%%%%%%%%%------------------------------
%%%%%%%%%------------------------------

\begin{restatable}{lemma}{decompToMeasure}
~\label{lemma:decompToMeasure}
Given a $(c_0,C)$-colourful Rabin graph $\Gc$ on which we have a $(c_0,C)$-colourful decomposition $\Dc$, there is an $\labelset$-labelled $(c_0,C)$-colourful tree with Rabin measure for $\Gc$, where no vertex is mapped to $\top$.
\end{restatable}

\begin{proof}[Proof of Lemma~\ref{lemma:decompToMeasure}]
The following is proved by induction on the size of $C$. Given a decomposition, we inductively obtain a tree and a corresponding mapping into the tree. 
We modify both the tree and Rabin measure thus obtained from the recursively defined decompositions and then merge them together. We later prove that indeed such a mapping defined is a Rabin measure.

\paragraph*{Suppose $C = \emptyset$} and the $(c_0,\emptyset)$-colourful decomposition $\Dc = \seq{V}$. Let $t$ denote the length of the longest path in $V$ that does not visit a vertex $v$ for which $c_0$ is a good colour. We consider an $\labelset$-labelled $(c_0,C)$-colourful tree obtained from the prefix closure $\Lc$ of the set $\{\tpl{\alpha_1\cdot\bot},\dots,\tpl{\alpha_t\cdot\bot}\}$, where $\alpha_1<\alpha_2<\dots<\alpha_t$, each $\alpha_i$ is an element of $\Nb$.

    All vertices $v\in V$ such that $c_0\in G_v$ are mapped to the empty sequence denoted by $\tpl{}$. For vertices $v$ where $c_0\notin G_v$, we define $D(v)$ to be the length of the longest path from $v$ to a vertex $u$ such that $c_0\notin G_u$. 
    We then define $\mu$ for each $v$ where $c_0\notin G_v$ and where $D(v) = i$ is finite, to be $(\alpha_i\cdot\bot)$s. Suppose $\mu(u_1) = \tpl{\alpha_{i}\cdot\bot}$ and $\mu(u_2) = \tpl{\alpha_j\cdot\bot}$ and if $D(u_1)<D(u_2)$ then $\alpha_i<\alpha_j$.

    To verify that such a mapping satisfies the condition, we know that since $B_v = \emptyset$ for all $v$, all the vertices satisfy \Rrr~trivially. 

    Now we consider all edges $u\rightarrow v$ and show that the edge satisfies \Gone~or \Gtwo. Notice that if $c_0\in G_u$, then it satisfies \Gtwo.   
    Else, $\mu(u)\succ\mu(v)$, since it must be the case that $D(u)>D(v)$, since the distance to any vertex which $c_0$ as a good colour from $u$ is at least one more than this distance from $v$. Therefore if $\mu(u) = \tpl{\alpha_j\cdot\bot}$ and $\mu(v) = \tpl{\alpha_i\cdot\bot}$, then $\alpha_j>\alpha_i$ and therefore $\mu(u)\succ\mu(v)$.

    \paragraph*{Suppose $C\neq \emptyset$} we have a $(c_0,C)$-colourful decomposition where $C\neq \emptyset$, where $$\Dc = \seq{A, (c_1,V_1,\Dc_1,A_1),\dots,(c_j,V_j,\Dc_j,A_j)}$$
    
    Then for each $V_i$, since it has a $(c_i,C\setminus\{c_i\})$-colourful decomposition, by induction, we have a mapping $\mu_i$ to an $\labelset$-labelled $(c_i,C\setminus\{c_i\})$-colourful tree $\Lc_i$.

    We give a Rabin measure $\mu$ into the tree $\Tc$ below
    $$\set{\tpl{}}\cup\eset{\tpl{\alpha^0_1\cdot\bot},\dots,\tpl{\alpha^0_t\cdot\bot}}\bigcup_{i=1}^j \eset{\tpl{\alpha^i_0\cdot c_i}\odot\Lc_i,\tpl{\alpha^i_1\cdot\bot},\dots,\tpl{\alpha^i_t\cdot\bot}}$$
%    \thejaswiniside{Will it help here to draw a picture? The tree is quite big}
    where $\alpha^i_\ell$ are elements from $\labelset$ such that $\alpha^{i_1}_\ell<\alpha^{i_2}_{\ell'}$ if $i_1<i_2$ and $\alpha^{i}_{\ell_1}<\alpha^{i}_{\ell_2}$ if $\ell_1<\ell_2$.

We define $\mu(u)$ from a decomposition $\Dc$ above as follows.
\begin{itemize}
    \item If $u\in A$ and  $c_0\in G_u$, then $\mu(u) = \tpl{}$.
    \item If $u\in A\setminus \{v\in A\mid c_0\in G_v\}$, we define $\mu(u) = \left(\alpha^{0}_\ell\cdot \bot\right)$ where $\ell$ is the length of the largest path from $u$ to a vertex $v$ such that $c_0\in G_v$.
    \item For vertices  $u\in V_i$, we define $\mu(u) = (\alpha^i_0\cdot c_i)\odot \mu_i(u)$.
    \item For vertices $u\in A_i\setminus V_i$, we define $\mu(u) = (\alpha^i_\ell\cdot \bot)$ where $\ell$ is the length of the largest path from $u$ to a vertex in $V_i$.
    \end{itemize}

We show that the $\mu$ defined above satisfies the conditions required for it to be a Rabin measure. For this we need to show each edge in the graph $\Gc$ is consistent.

Before this we use the following observation about the mapping defined. For the rest of the proof,  we denote $A_0 = A$, $V_0 = \{v\in A\mid c_0\in G_v\}$. Let $W_i$ be defined similarly to the definition of a decomposition, where $W_1 = V\setminus A$, and $W_{i+1} \:= W_i\setminus A_i$. We moreover define $W_0 \:=V$. 
\begin{quote}{$(*)$}
    For $i\in \{0,1,\dots,j\}$, any vertex in $u \in V\setminus W_{i}$ is such that $\mu(u)\prec \mu(v)$ for any $v\in W_i$.
\end{quote}

\begin{itemize}
    \item If $u\in A$ and $c_0\in G_u$, since $\mu(u) = \tpl{}$, for such a vertex any edge $u\rightarrow v$ satisfies \Gtwo, since the root is coloured with $c_0$, and also satisfies \Rrr~since $c_0\notin B_u$.
    
        \item If $u\in A_i\setminus V_i$, for $i\in \{0,1,\dots,j\}$, and suppose $\mu(u) = (\alpha^i_\ell\cdot \bot)$ we show that edges from $u$ satisfies \Gone. 
    \begin{itemize}
        \item If $v\in V\setminus W_i$, then edge $u\rightarrow v$ satisfies \Gone~from $(*)$, as all vertices in $V\setminus W_i$ are mapped to a node strictly smaller than $\mu(u)$ already. 
        \item If $v\in A_i$, all paths from $A_i$ in $W_i$ (as defined in the definition of a decomposition) leads to a vertex in $V_i$. 
        \item If $v\in V_i$ then by definition it is mapped to a descendent of $\alpha^i_0$ and is therefore mapped to a value smaller than $(\alpha^i_\ell\cdot \bot)$, and satisfies \Gone. If not, then $v$ is a neighbour  of $u$ in $A_i\setminus V_i$ and must have a distance (in $W_i$) to a vertex in $V_i$ to be strictly smaller than that from $u$. 
    \end{itemize}
    Therefore for any neighbour $v$, from our assignment of $\mu$, it must be the case that $\mu(v) = (\alpha^i_{\ell_1}\cdot \bot)$, where $\ell_1<\ell$, and hence $\mu(u)\succ \mu(v)$. 
    Observe that all edges from $u$ also satisfies \Rrr~because the only ancestor of $\mu(u)$ is $\tpl{}$, and it is coloured with $c_0$, and $c_0\notin B_v$ for any $v$, and therefore specifically $c_0\notin B_u$.

    \item If $u\in V_i$ for $i\in \{1,\dots,j\}$, for all edges $u\rightarrow v$, $v$ is either in in $V_i$ or in $V \setminus W_i$ since there are no paths from $V_i$ to $W_i \setminus V_i$. If $v\in V\setminus W_i$, we know $\mu(u) \succ \mu(v)$ from $(*)$, and thus \Gone~is satisfied. On the other hand if $v \in V_i$, then $\mu_i(u)$ and $\mu_i(v)$ are both defined. If edge $u\rightarrow v$ satisfies ~\Gtwo~with respect to $\mu_i$, then it continues to be satisfied in $\mu$ since  $\colouring(\mu(u)) = \colouring(\mu_i(u))$. Otherwise, the edge $u\rightarrow v$ satisfies~\Gone~in  $\mu_i$, i.e. $\mu_i(u) \succ \mu_i(v)$. Then $\mu(u) \succ \mu(v)$ since $\mu$ appends the same value to the beginning of $\mu_i(u)$ and $\mu_i(v)$. Thus, \Gone~is satisfied with respect to $\mu$ as well.

    Observe that  $\accumulatedcolourset(\mu(u)) = \accumulatedcolourset(\mu_i(u)) \cup \{c_i\}$. Also notice that from the definition of a decomposition, $c_i\notin B_u$ for any $u\in V_i$. So, if $\accumulatedcolourset(\mu_i(u))\cap B_u = \emptyset$, then $B_u\cap \accumulatedcolourset(\mu(u)) = \emptyset$. Thus, \Rrr~is also satisfied by edge $u\rightarrow v$.
\end{itemize}
\end{proof}

%%%%%%%%%%%%%------------------------------
%%%%%%%%%%%%%------------------------------
For the proof of Lemma~\ref{lemma:measureeasy} which would show how a Rabin measure serves as a witness that all infinite paths in a Rabin graph satisfy the Rabin condition, we require the following two simple facts on trees. These hold in general for all ordered trees and not just colourful ordered trees.  First one in Proposition~\ref{prop:irmaksprop} says that all the ancestors of a larger node in the tree is always either an ancestor of a smaller node or is also larger than the smaller node. 

The latter proposition is about an infinite sequence of nodes in a tree where two consecutive nodes satisfy some given properties.

%n infinite sequence of elements from the tree that satisfy the following properties.

\begin{restatable}{proposition}{irmakprop}
\label{prop:irmaksprop}
Any ancestor $t$ of $t'$ is such that for any other node $t''\prec t'$, either $t$ is  an ancestor of $t''$ or $t$ is strictly larger than $t''$.
\end{restatable}
\begin{proof}[Proof of Proposition~\ref{prop:irmaksprop}]
Assume $\ancestor(t'', t)$ is different from $t$. Then it is a strict ancestor of $t$. Since $t'' \prec t'$ and $t$ is an ancestor of $t'$,  either $t'' = \ancestor(t, t'')$ or $t'' \prec_{|t|} t=_{|t|} t'$. Here $\prec_a$ (resp. $=_a$) compares the first $a-$many elements of two tuples in accordance with $\prec$ (resp. $=$). This says that, since $t$ and $t'$ have the same first $|t|-$many entries, and $t$ is not an ancestor of $t''$ (i.e. $t'' \neq_{|t|} t$), $t''$ should either be an ancestor of $t$, or smaller than $t$ in the first $|t|-$many entries. In both of these cases we get $t'' \prec t$.
\end{proof}
%%%%%%%----------------------------------------
%%%%%%%-----------------------------------------
\begin{proposition}[Lemma 1, \cite{KK91}]~\label{prop:smalltreeproposition}
Consider an infinite sequence $\rho$ of nodes from $\Lc$, an $\labelset$-labelled $(c_0,C)$-colourful tree, where $\rho = t_0,t_1,\dots,t_i,\dots$. Suppose for all $j\in \mathbb{N}$, if
\begin{itemize}
    \item either $t_j\succ t_{j+1}$  or
    \item $t_j$ is an ancestor of $t_{j+1}$
\end{itemize}  then the smallest element of the sequence, denote ed by $t_{\inf}$  must be
\begin{enumerate}
    \item  the largest common ancestor of $t_i$ and $t_{i+1}$ infinitely often
    \item  an ancestor of all but finitely many $t_i$s.
\end{enumerate}
\end{proposition}
\begin{proof}
Let $p$ be the position after which all $t_k$ such that $k>p$ are such that $t_k\in \inf\{\rho\}$. Without loss of generality, assume $t_p = t_{\inf}$. Clearly, $t_{p+1}\succeq t_p$, since it is the smallest among $\inf\{\rho\}$. 
\begin{enumerate}
    \item We recall that $t_p\succ t_{p+1}$ or $t_p$ is an ancestor of $t_{p+1}$. And we can conclude that $t_p$  is an ancestor of $t_{p+1}$. Since after position $p$, each element occurs infinitely many times, we have that $t_{\inf}$ is the largest common ancestor of $t_{\inf}$ and all its occurrences $t_i$ and its successors $t_{i+1}$.
    \item We also argue that $t_{\inf}$ is an ancestor of all $t_{j}$ for $j\geq p$. Let the next occurrence of $t_{\inf}$ in $\rho$ be at $t_q$, where $q>p$. We will show that for all $p\leq j\leq q$, $t_{\inf}$ is an ancestor of $t_j$, or equivalently that $t_p$ is an ancestor of $t_j$. Indeed, consider $t_p,t_{p+1},\dots,t_q$.

    We show $t_p = \ancestor(t_p,t_p)= \ancestor(t_p,t_{p+1})=\dots=\ancestor(t_p,t_{q}) = t_p$. We proceed by induction. In the base case, trivially $t_p $ is an ancestor of $t_p$.
    We assume as the induction hypothesis that $t_p$ is an ancestor of $t_{p+i}$. We know that $t_{p+i}$ and $t_{p+i+1}$ satisfy either  $t_{p+i}\succ t_{p+i+1}$ or  $t_{p+i}$ is an ancestor of $t_{p+i+1}$. 

    In the latter case $t_{p+i}$ is an ancestor of $t_{p+i+1}$. By the inductive hypothesis we have that $t_p$ is an ancestor of $t_{p+i}$. Therefore, we conclude that $t_p$ is an ancestor of $t_{p+i+1}$.
 
    In the former case, we invoke Proposition~\ref{prop:irmaksprop} with $t:= t_p$, $t':=t_{p+i}$ and $t'':=t_{p+i+1}$. We consequently get either $t_{p+i+1} \prec t_p$ or $t_p$ is an ancestor of  $t_{p+i+1}$. Since $t_p = t_{\inf}$, this gives us $t_p$ is an ancestor of $t_{p+i+1}$ concluding our claim. 
\end{enumerate}
\end{proof}

%%%%%%------------------------------------
%%%%%%------------------------------------

\begin{restatable}{lemma}{measureeasy}
~\label{lemma:measureeasy}
If there is an $\labelset$-labelled $(c_0,C)$-colourful Rabin measure for a $(c_0,C)$-colourful Rabin graph $\Gc$ and no vertex is mapped to $\top$, then all infinite paths in the graph $\Gc$ satisfy the Rabin condition. 
\end{restatable}

\begin{proof}[Proof of Lemma~\ref{lemma:measureeasy}]\label{prf:measureeasy}
Consider an infinite path $\pi = v_0\rightarrow v_1 \rightarrow \dots \rightarrow v_j\rightarrow\dots$ in $\Gc$. We define the infinite sequence $\mu(\pi)\: = \mu(v_0),\mu(v_1),\ldots,\mu(v_j),\ldots$, obtained by taking the image of the run on the colourful tree. In this colourful tree, consider $t = \min\{\inf(\mu(\pi))\}$, and let $c = \colouring(t)$. 

For such a $t$, we show 
\begin{enumerate}
    \item $t$ is not coloured with $\bot$;~\label{en:easynotbot}
    \item $c\in G_v$, for infinitely many $v$ in $\pi$, and ~\label{en:easyG}
    \item $c\notin B_v$ for each $v$ occurring after some finite prefix in $\pi$,~\label{en:easyR}
\end{enumerate}
to conclude that $\pi$ satisfies the Rabin condition.
Before we begin the rest we first remark that from conditions \Gone~or \Gtwo, we get that $(v_{j},v_{j+1})$ is such that either one of the following is true, either $\mu(v_{j})\succ \mu(v_{j+1})$, or $\mu(v_{j}) = \ancestor(\mu(v_{j}), \mu(v_{j+1}))$. 

Therefore $t = \min\{\inf(\mu(\pi))\}$ must be the largest common ancestor of $(\mu(v_{i}), \mu(v_{i+1}))$ infinitely often, and moreover, must be a common ancestor of $\mu(v_{j})$ for almost all $j$.

To show~\ref{en:easynotbot} consider a vertex $v_i$, and $(v_{i},v_{i+1})$ which occurs infinitely often in the play $\pi$ for which $\mu(v_i)= t$ and also 
where the edge is consistent. This especially means that this edge satisfies condition \Gone~or \Gtwo. If $v_i$ is coloured with $\bot$, this edge can only satisfy \Gone, and hence $\mu(v_{i})\succ\mu(v_{i+1})$, a contradiction to the assumption that $\mu(v_i) = \liminf(\mu(\pi))$.

Item~\ref{en:easyG} which claims that $c\in G_v$ infinitely often for vertices from the play $\pi$ also follows from the above conditions as edge $(v_{i},v_{i+1})$ identified in the above condition should satisfy \Gtwo infinitely often where $\mu(v_i) = t$.

Finally, we show item~\ref{en:easyR} that $c\notin B_v$ for any $v$ after some finite prefix of $\pi$. This is because for any $(v_{j},v_{j+1})$, where we have  $c\notin \accumulatedcolourset(\mu(v_{j}))$ from condition \Rrr. Since we had earlier observed that $t$ is a common ancestor of $\mu(v_{j})$ 
 for all but finitely many of the edges $(v_{j},v_{j+1})$ in $\pi$, we must have $c\notin B_{v_j}$ for all but finitely many $v_j$s. 

This proves the above Lemma and therefore the (\ref{en:measure} $\implies$ \ref{en:one}) of Theorem~\ref{thm:measure}.
\end{proof}

%%%%%%%%%%%------------------------------------
%%%%%%%%%%%------------------------------------

The following lemma is not necessary in proving the equivalences but is provided to help the understanding of the equivalence of Item~\ref{en:decomp} and Item~\ref{en:one} in Theorem~\ref{thm:measure} more clearly. 
\begin{restatable}{lemma}{measureeasydecomp}
~\label{lemma:measureeasydecomp}
If there is an $\labelset$-labelled $(c_0,C)$-colourful decomposition for a $(c_0,C)$-colourful Rabin graph $\Gc$, then all paths satisfy the Rabin condition in $\Gc$.
\end{restatable}

\begin{proof}[Proof of Lemma~\ref{lemma:measureeasydecomp}]\label{prf:measureeasydecomp}
    If $\Dc = \seq{V}$, then notice that all infinite paths from all vertices in $V$, visit a vertex $v$ such that $c_0\in G_v$, and moreover, by assumption that $\Dc$ is a $(c_0,C)$-colourful decomposition, and hence $c_0\notin B_v$ for any $v\in V$.
    This also means that all infinite paths in $V$ satisfy Rabin condition for colour $c_0$.

    If $\Dc = \seq{A, \tpl{c_{1},V_1, \Dc_1,A_1},  \tpl{c_{2},V_2, \Dc_2,A_2}, \dots, \tpl{c_{j},V_j,\Dc_j,A_j}}$, any play $\pi$ such that $\inf\{\pi\}\subseteq V_i$, then we know that $\pi$ satisfies the Rabin condition by induction on the size of the decomposition. 
    Indeed, no vertex in $V_i$ is such that $c_i\in B_v$, and $V_i$ has a $(c_i,C\setminus \{c_i\})$-colourful decomposition $\Dc_i$.
    
     Observe that any path with $\inf\{\pi\}\cap A\neq \emptyset$ satisfies the Rabin condition with colour $c_0$ since $c_0\notin B_v$, and all infinite paths from $A$ is such that there is some vertex $v$ with $c\in G_v$. 
     If on the other hand there are values $i_1\neq i_2$ such that $\inf\{\pi\}\cap A_{i_1}\neq \emptyset$ and $\inf\{\pi\}\cap A_{i_2}\neq \emptyset$, then we claim that we also must have $\inf\{\pi\}\cap A\neq \emptyset$, and hence this path satisfies  the Rabin condition.      
The previous statement can be shown with a simple argument about the structure of the decomposition.  We declare $A_0$ and $V_0$ to both denote the set $A$ and $W_{i}$ defined as in the definition of a colourful decomposition. Let $i\in \{0,1,2,\dots, j\}$ be the smallest value such that $\inf\{\pi\}\cap A_{i}\neq \emptyset$. If $i=0$, we have proved our above claim. If $i>0$, then after some prefix, no vertices from $A_\ell$ occur in $\pi$ for $\ell<i$. But we know that in $W_i$, there are no paths to vertices in $W_i\setminus V_i$, and therefore, after some finite point of the infinite path, we must have $\inf\{\pi\}\subseteq V_i$. 
\end{proof}

  \section{Appendix for Section~\ref{sec:unilift}}
  
%%%--------------------------------
%%%--------------------------------
 \minoftwo*
\begin{proof}[Proof of~\cref{lemma:minoftwo}]\label{prf:minoftwo}
We show that for two measures $\mu_1\preceq \mu_2$, 
 that $\biglift_v(\mu_1)\preceq \biglift_v(\mu_2)$. Note that it suffices to show that for \Controller's (resp. \Environment's) vertices $v$, the value $t_1 = \min_{(v,w) \in E}\{\lift_{\mu_1}(v,w)\}$ is at most as large as  $t_2 = \min_{(v,w) \in E}\{\lift_{\mu_2}(v,w)\}$ (using $\max{}$ for \Environment~instead). We instead argue that $\mu_1'$, defined as $\mu_1[v := t_2]$  ensures that the vertex $v$ is consistent. Recall that a \Controller~vertex is consistent if it has one consistent outgoing edge and an \Environment~vertex is consistent if all its outgoing edges are consistent. Since by definition $t_1$ is the smallest element larger than $\mu_1(v)$ that makes some (resp. all) $v \rightarrow w $ consistent in $\mu_1[v := t_1]$, this gives us $t_1 \preceq t_2$. 
%Let $w$ be the edge such that $v\rightarrow w$ is consistent in $\mu_2' = \biglift_v(\mu_2)$. We argue that the same edge(s) $v\rightarrow w$ is still consistent in $\mu'$. 
Let $v\rightarrow w$ be an outgoing edge of $v$ that is consistent in $\mu_2' = \biglift_v(\mu_2)$. We claim that it is consistent in $\mu_1'$ as well.
    \begin{itemize}
        \item If $v \rightarrow w$ satisfied \Gone~with respect to $\mu_2'$, then it continues to satisfy \Gone~with respect to $\mu_1'$, since $\mu_1'(v) = \mu_2'(v) \succ \mu_2(w)\succeq
        \mu_1'(w)$. 
        \item  If $v \rightarrow w$ satisfied \Gtwo~with respect to $\mu_2'$, then it either continues to satisfy \Gtwo, or satisfies \Gone~with respect to $\mu_2'$.
        To see this, we observe that $\mu_2'(v)$ is an ancestor of  $\mu_2'(w) = \mu_2(w)$, and $\mu_2'(w)\succeq \mu_1'(w)$. From Proposition~\ref{prop:irmaksprop}, we consequently get either $\mu_1'(v) \succ \mu_1'(w)$ or $\mu_1(v)$ is an ancestor of $\mu_1'(w)$, which is exactly \Gone~or \Gtwo~respectively. Additionally $\colouring(\mu_1'(v)) \in G_v$ is trivially satisfied since $\mu_1'(v) = \mu_2'(v)$ and $ \colouring(\mu_2'(v)) \in G_v$.    
        \item If $v \rightarrow w$ satisfied \Rrr~with respect to $\mu_2'$, then it continues to satisfy \Rrr~with respect to $\mu_1'$, since $\mu_2'(v) = \mu_1'(v)$ and $\accumulatedcolourset(\mu_1'(v))\cap B_v = \accumulatedcolourset(\mu_2'(v))\cap B_v = \emptyset$.
    \end{itemize}
\end{proof}

\bigLift*

% \begin{restatable}{prop}{smalllift}
% ~\label{prop:smalllift}
%     Given a mapping from the vertices of an $n$-vertex $(c_0,C)$-colourful Rabin game $\Gc$ with $m$ edges and $k$ colours into the $\labelset$-labelled $(c_0,C)$-colourful tree $\Lc$, computing the function $\lift_\mu{(u,v)}$ takes time at most $O\left( T_{\mynext}\right)$, where $T_{\mynext}$ is defined as the maximum of
%     \begin{itemize}
%         \item the time taken to make a linear pass on a node in $\Lc$;
%         \item the time taken to compute the next node in $\Lc$.
%         \item given $t\in \Lc$ and $C'\subseteq C$ such that  $\colouring(t) \in C'$, the time taken to find the next node that uses colours only from $C'\cup \{\bot\}$. 
%     \end{itemize}
% \end{restatable}

\begin{proof}[Proof of~\cref{lemma:bigLift}]\label{prf:bigLift}

To prove the above Lemma, we first %state a small proposition about $\lift_\mu(u,v)$.
answer the following question: given an edge $u\rightarrow v$ and a mapping $\mu$ to $\Lc^\top$, can we calculate $\lift_\mu(u,v)$? 

We show how using the following subroutines mentioned: (1) computing the next node: we denote the successor of $t$ in $\Lc^\top$ with respect to the order $\prec$ by $\mynext(t)$
(2)~given $t\in \Lc$ and $C'\subseteq C$ such that  $\colouring(t) \in C'$, finding the next node whose colour set contains colours only from $C'\cup \{\bot\}$.

A naive way to compute $\lift_\mu{(u,v)}$ would be to apply $\mynext$ to $\mu(u)$ and to check each time if the edge $u \rightarrow v$ satisfies the consistency properties. But such an algorithm would potentially take exponential time for computing some lift functions.
We remark however that this n\"{a}ive algorithm would only add a polynomial factor to the upper bound to the worst case complexity of our run-time after amortisation. 

We will now give the function which directly computes $\lift_\mu(u,v)$ using only two primitives stated above after a linear scan: (a) finding the next colour in the tree;  (b) given a $C'\subseteq C$ and a $t$ such that $\colouring(t)\in C'$, finding the next node $t'$ of $t$ which satisfies  $\accumulatedcolourset(t') \subseteq C'\cup \{\bot\}$.

The details of the computation can be inferred from the procedure described below, as we need only finitely many linear passes on a node's description. 

\paragraph*{Edge $u \rightarrow v$ is already consistent} In this case,  $u \rightarrow v$ already satisfies at least one of \Gone~or \Gtwo~along with (R) in $\mu$. Hence  $\lift_\mu(u,v)$ is set to $\mu(u)$, continues to make $u \rightarrow v$ consistent.
%\paragraph*{Edge $(u,v)$  satisfies (G2) but not (R) or (G1)}

\paragraph*{Edge  $u \rightarrow v$ satisfies \Gone~but not \Rrr}
In this case, we only need to find the smallest value larger than $\mu(u)$ whose colour set
does not contain any colours from $B_u$. 
Let $\mu(u) \:=(x_1\cdot c_{i_1},\dots,x_m\cdot c_{i_m})$. 
%and $\mu(v)\:= (y_1\cdot c_{j_1},\dots,y_m\cdot c_{j_p})$.
We achieve this by finding the largest position $s$ that gives $\accumulatedcolourset((x_1\cdot c_{i_1},\dots,x_s\cdot c_{i_s})) \cap B_u = \emptyset$.
%Observe that this implies $c_{i_{s+1}}\in B_u$ and $c_{i_{s+1}}\neq \bot$.
Then we compute the smallest child $t$ larger than the node above, $t = (x_1\cdot c_{i_1},\dots,x_{s+1}\cdot c_{i_{s+1}})$ that gives $\accumulatedcolourset(t)\cap B_u = \emptyset$ and set $\lift_\mu(u,v)$ to $t$. The computation clearly takes time at most $T_\mynext$.

Since $\mu(u) \succ \mu(v)$, $\lift_\mu(u,v) \succ \mu(u)$ and $\lift_\mu(u,v)$ doesn't use any colours from $B_v$, the edge $u \rightarrow v$ satisfies \Gone~and \Rrr~in the new mapping.
\paragraph*{Edge $u \rightarrow v$ satisfies \Gtwo~but not \Gone~or \Rrr}
We again take $\mu(u) \:=(x_1\cdot c_{i_1},\dots,x_m\cdot c_{i_m})$. Since $u \rightarrow v$ satisfies \Gtwo, we know that $\mu(u)$ is an ancestor of $\mu(v)$. We argue that the smallest value larger than $\mu(u)$ that also satisfies \Rrr~does not satisfy \Gtwo, but rather satisfies \Gone. This is because there is an ancestor of $\mu(u)$ (and thus, of $\mu(v)$) that is coloured by a bad colour of $u$. Since $\lift_\mu(u,v)$ must be larger than $\mu(u)$, it cannot be set to an ancestor of $\mu(u)$. Then, it should be set to a larger sibling of one of the ancestors of $\mu(u)$. Since any larger sibling of an ancestor of $\mu(v)$ is always larger than $\mu(v)$, the smallest value of $\lift_\mu(u,v)$ that makes $u \rightarrow v$ consistent satisfies \Gone. We have therefore reduced this case to the previous one. 

%This means that $\lift_{\mu}(u,v)$ needs to be at least $\mynext_C(\mu(v))$, and we have therefore reduced this to a case where $(u,v)$ satisfies \Gone.
%\irmak{(Alternative) Instead of the last half-sentence provide the method to get the next node that makes $(u,v)$ consistent explicitly with: We can therefore conclude that $\lift_\mu(u,v)$ provided in the previous case makes $(u,v)$ consistent as well.}
%Let $\mu(u) :=(x_1\cdot c_{i_1},\dots,x_m\cdot c_{i_m})$ and $s$ the largest position $i_s$, such that  $c_{i_s}\notin B_u$. 
   %$$\lift_\mu(u,v) = \min\left(\min_{c\in G_u}\left(\mynext_{C'\setminus\{c\}}^\ell(x|_s)\right), \mynext_{C'}^\ell(x|_s)\right)$$\thejaswini{This is not done, this does not satisfy G2 need to add right colour at end}
\paragraph*{Edge  $u \rightarrow v$ satisfies neither \Gone, \Gtwo~or \Rrr}
Since the edge does not satisfy \Gone, we know $\mu(u)\preceq \mu(v)$. We go through the ancestors of $\mu(v)$ one by one in increasing order to see if there exists one that is both strictly larger than $\mu(u)$, and satisfies \Rrr. %in increasing order, if setting $\mu(u)$ to this value satisfies \Gtwo~and \Rrr.
If there exists one, then we set $\lift_\mu(u,v)$ to the first such value found, and $u \rightarrow v$ satisfies \Gtwo~and \Rrr~in the new mapping. 
This computation takes a linear scan through at most the length of $\mu(v)$. If none of the ancestors satisfy these constraints, then we know that $\lift_\mu(u,v)$ has to be at least as large as  $\mynext(\mu(v))$. Thus $u \rightarrow v$ has to satisfy \Gone~and \Rrr~in the next mapping. Once more, we have reduced this case to the previous ones. %Note that $\mynext(\mu(v))$ is indeed the smallest value that makes $u \rightarrow v$ consistent, since any ancestor of $\mu(v)$ is smaller than the next node of $\mu(v)$. 
%\irmak{
%Let $ \mu(u):=(x_1 \cdot c_{i_1}, \dots, x_m \cdot c_{i_m})$ and $ \mu(v):=(x'_1 \cdot c'_{i_1}, \dots, x'_{m'} \cdot c'_{i'_{m'}})$.
%Since $\mu(u) \prec \mu(v)$ and $\mu(u)$ is not an ancestor of $\mu(v)$, there exists some $n$ that satisfies $\mu(u) \prec_n \mu(v)$. Take $n$ to be the smallest such number that additionally satisfies $\{c_1, \dots, c_{n-1}\} \cap B_u = \emptyset$}

We have concluded that computing $\lift_{\mu}(v,w)$ takes time at most $O\left( T_{\mynext}\right)$.
Recall the definition of $\biglift_v$ by using $\lift_{\mu}(v,w)$ as a subroutine.
    $$\biglift_v(\mu)(u) = \begin{cases}
        \mu(u)  & \text{for }u\neq v\\
        \min_{(v,w)\in E}\left\{\lift_\mu(v,w)\right\} & \text{if } u=v\in V_c\\
        \max_{(v,w)\in E}\left\{\lift_\mu(v,w)\right\} & \text{if } u=v\in V_e
\end{cases}$$
 It is therefore easy to conclude that $\biglift_v(\mu)(v)$ takes time at most $O\left(\deg(v) \cdot T_{\mynext}\right)$

\end{proof}

%%%%%%%%%%%%%%------------------------------------
%%%%%%%%%%%%%%------------------------------------

%\liftingalgo* 

  \section{Appendix for Section~\ref{sec:colourfulTrees}}
  \begin{restatable}{proposition}{colourfuluniversal}\label{prop:colourfuluniversal}
    The $(c_0,C)$-colourful tree $\Uc_{(c_0,C)}^\ell$, embeds any $(c_0,C)$-colourful tree with at most $n$ leaves where $\ell \geq \lceil \log n\rceil$.
\end{restatable}
\begin{proof}[Proof of Proposition~\ref{prop:colourfuluniversal}]\label{prf:colourfuluniversal}
    Consider any $(c_0,C)$-colourful tree $\Tc$ with $n$ leaves. 
    The statement is trivial if $C = \emptyset$, since from our construction, our universal tree is such that all $n$ leaves have colour $\bot$.
    We assume $C\neq \emptyset$ but $\ell = 0$, and therefore $n = 1$. 
    Let $C = \{c_1,\dots,c_h\}$. In this case, we have 
    $$\Uc^0_{(c_0,C)} \:= \seq{\Uc_{(c_1,C_1)}^0,\dots,\Uc_{(c_h,C_h)}^0,\left(\bot,\seq{}\right)}$$
    We must either have $\Tc = \tpl{c_0,\seq{\Tc_i}}$ for some $(c_i, C\setminus\{c_i\})$-colourful tree $\Tc_i$ or alternatively, $\Tc = \seq{\left(\bot,\seq{}\right)}$, and clearly from the construction, it follows that this tree can be embedded in $\Uc_{c_0,C}^0$, recursively, by choosing an appropriate subtree $\left(c_i,\Uc_{(c_i,C_i)}^0\right)$, and recursively embedding $\Tc_i$ in $\Uc_{(c_i,C_i)}^0$.

    If we consider the case where  
    $n>1$ (and therefore $\ell>0$), and suppose our tree with $n$ leaves is $\Tc = \tpl{c_0,\seq{\Tc_1,\dots,\Tc_m}}$.
    Let $n_p$ represent the number of leaves of $\Tc_{p}$. We know $\sum n_p = n$. For each $p$, we define $$\Tc_{<p} = \seq{\Tc_1,\dots,\Tc_{p-1}}$$ and $$\Tc_{>p} = \seq{\Tc_{p+1},\dots,\Tc_m}.$$  
    
    There must be at least one $p$ for which $\Tc_{<p}$ as well as $\Tc_{>p}$ has size at most $n/2$. The existence of such a $p$ can be shown by defining a summation $N_j = \sum_{i=1}^j n_i$ which ranges from $0$ to $n$ as $j$ ranges from $1$ to $m$. Then there must be some point where $N_j$ exceeds $n/2$, giving us our desired $p$.
    
     Since both $\Tc_{<p}$ and $\Tc_{>p}$ have at most $n/2$ leaves, by induction $\Uc^{\ell-1}_{(c_0, C)}$ embeds $\Tc_{<p}$ as well as $\Tc_{>p}$, since $C$ contains all the colours in $\Tc_{<p}$ and $\Tc_{>p}$ and each tree has less than $n/2$ leaves. Furthermore,
 $\Uc_{(c_{i_p},C_{i_p})}^\ell$ embeds $\Tc_p$, where $c_{i_p}$ is the colour of the root of $\Tc_p$ and $C_{i_p} = C\setminus \{c_{i_p}\}$. Observe that for each $c_i$, there is a copy of the tree of $\Uc_{(c_i,C_i)}^\ell$, where $C_i = C\setminus \{c_i\}$.
Hence from the construction of $\Uc_{(c_0,C)}^\ell$, the tree $$\Tc = \Tc_{<p}\cdot\tpl{c_0,\seq{\Tc_p}}\cdot \Tc_{>p}$$ can be embedded into 
     $$\Uc^\ell_{(c_0,C)} \:= \Uc_{(c_0,C)}^{\ell-1}\cdot
     \tpl{c_0,\seq{\Uc_{(c_1,C_1)}^\ell,\dots,\Uc_{(c_h,C_h)}^\ell, \tpl{\bot,\seq{}}}}
     \cdot\Uc_{(c_0,C)}^{\ell-1}.$$ 
\end{proof}

%%%%%%%%%-------------------------------
%%%%%%%%%-------------------------------
\begin{restatable}{lemma}{smallcolourful}
~\label{lemma:smallcolourful}
    The tree $\Uc_{(c_0,C)}^\ell$ has at most $2^k\cdot k!\cdot 4^\ell$ many leaves where  $k = |C\cup \{c_0\}|$.
\end{restatable}

\begin{proof}[Proof of Lemma~\ref{lemma:smallcolourful}]\label{prf:smallcolourful}
Let us denote by $U(\ell,h)$, the  number of leaves in the tree $\Uc_{(c_0,C)}^\ell$ defined above, where $|C| = h$. 

%For $t\geq 1$ and $k\geq 1$, where $t = \lceil\lg(n)\rceil$ then
If $k = 1$, $h= k-1 = 0$ then $U(\ell,h) = 2^\ell$ by construction.

If $\ell = 0$, then we show by induction a stronger statement that $U(0,h)\leq h!h$ for all values of $h\geq 1$.
Indeed, 
 $$\Uc^0_{(c_0,C)} \:= \seq{\Uc_{(c_1,C_1)}^0,\dots,\Uc_{(c_h,C_h)}^0,\left(\bot,\seq{}\right)}$$
From this we can infer that 
$$U(0,h) \leq (k-1)\cdot U(0,k-1) + 1$$
Since we already know $U(0,1) = 1$, inductively, we can show that %assume that $U(0,k) \leq (k-1)!(k-1) $, for $\ell = 0$ and $k<h$, then 
$$U(0, h)\leq h\cdot U(0,h-1) + 1\leq h\cdot\left((h-1)!\right)+1\leq h!h$$

For $\ell,h > 0$, recall that  
 $$\Uc^\ell_{(c_0,C)} \:= \Uc_{(c_0,C)}^{\ell-1}\cdot
     \tpl{c_0,\seq{\Uc_{(c_1,C_1)}^\ell,\dots,\Uc_{(c_h,C_h)}^\ell, \tpl{\bot,\seq{}}}}
     \cdot\Uc_{(c_0,C)}^{\ell-1}.$$ 
Therefore, we see that for $\ell,h>0$, the following recurrence relation holds
$$U(\ell,h) = 2\cdot U(\ell-1,h) + h\cdot U(\ell,h-1) + 1$$

We prove $U(\ell,h)\leq 4^\ell \cdot hh!\cdot 2^h$, by induction.

For the base case, we can see $U(0,h)$ and $U(\ell,0)$, the inequality holds.
We assume for $t<\ell$ and $j<h$, that $U(t,j)\leq 4^t \cdot  2^j\cdot jj!$ as our induction hypothesis.
For this $\ell$ and $h$, observe
\begin{align*}
U(\ell,h) & =  2U(\ell-1,h) + hU(\ell,h-1) + 1&\\
     & \leq 2\cdot\left( (4)^{\ell-1}\cdot 2^h \cdot (h)!h\right)+ h\cdot \left(4^\ell\cdot 2^{h-1}\cdot h!\right) + 1&\\
      & \leq \frac{1}{2}\left( 4^\ell\cdot 2^h \cdot h!h\right) + h\left(4^\ell\cdot 2^{h-1}\cdot (h-1)!(h-1)\right) + \left(4^\ell\cdot 2^{h-1}\cdot (h)!\right)&\\
      & =  \frac{1}{2}\left( 4^\ell\cdot 2^{h-1} \cdot h!h\right) + \left(4^\ell\cdot 2^{h-1}\cdot h!h\right) &\\
     & =  \frac{1}{2}\left( 4^\ell\cdot 2^{h-1} \cdot h!h\right) + \frac{1}{2} \cdot \left(4^\ell\cdot 2^{h}\cdot h!h\right)&\\
     & =  4^\ell\cdot 2^h \cdot (h)!h
\end{align*}
Since $h=k-1$, our claim follows.
\end{proof}

%%%%----------------------------------
%%%%----------------------------------

\begin{restatable}{lemma}{smallcolourfulcombibound}\label{lemma:smallcolourfulcombibound}
    The tree $\Uc_{(c_0,C)}^\ell$ has size at most ${\ell +k\choose k-1}\cdot2^\ell\cdot k!$, where $k = |C\cup \{c_0\}|$.
\end{restatable}

\begin{proof}[Proof of Lemma~\ref{lemma:smallcolourfulcombibound}]\label{prf:smallcolourfulcombibound}
Let us again denote by $U(\ell,h)$, the number of leaves in the tree $\Uc_{(c_0,C)}^\ell$, where $|C| = h = k-1$. 

If $h = 0$, then $U(\ell,h) = 2^\ell$ by construction, and therefore we have $U(\ell,0) = {\ell +h+ 1\choose h} h! 2^\ell$.

If $\ell = 0$, recall from the proof of Lemma~\ref{lemma:smallcolourful} we show that $U(0,h)\leq h!h$ for all values of $h\geq 0$.

Now, suppose $\ell,h>0$, then we have 
\begin{align*}
U(\ell,h) & =  2U(\ell-1,h) + hU(\ell,h-1) + 1&\\
     & \leq 2\cdot\left( (2)^{\ell-1}\cdot {\ell +h\choose h} (h)!h\right)+ h\cdot \left(2^\ell\cdot {\ell +h\choose h-1} (h-1)!(h-1)\right) + 1&\\
      & \leq \left( 2^\ell\cdot{\ell +h\choose h} (h)!h\right) + \left(2^\ell\cdot {\ell +h\choose h-1} (h)!(h-1)\right) + \left(2^\ell\cdot {\ell +h\choose h-1} (h)!\right)&\\
      & \leq   2^\ell\cdot h!h\left({\ell +h\choose h} + {\ell +h\choose h-1} \right) &\\
     & =   2^\ell\cdot h!h {\ell +h + 1\choose h} \leq 2^\ell\cdot k! {\ell +k\choose k-1}
\end{align*}
\end{proof}

%%%%%------------------------------------
%%%%%------------------------------------
%\lowerboundcolourfultree* 

%%%%%%-----------------------------------
%%%%%%-----------------------------------
\succinctcolourfullabeling* 

\begin{proof}\label{prf:succinctcolourfullabeling}
We use two different operators in the below construction to obtain a $\words$-labelling of $\Uc_{(c_0,C)}^\ell$: 
\begin{itemize}
    \item for $b \in \{0,1\}$ and $\omega_1\dots,\omega_m\in \words\cdot C$, we define $$b\cdot \Lc \: = \{(b\cdot \omega_1,\omega_2,\dots, \omega_m)\mid (\omega_1,\omega_2,\dots, \omega_m)\in \Lc\}$$ and
    \item for $\alpha$ and $\omega_1\dots,\omega_m\in \words\cdot C$, we define $$\alpha \odot \Lc \: = \{(\alpha,\omega_1,\omega_2,\dots \omega_m)\mid (\omega_1,\omega_2,\dots, \omega_m)\in \Lc\}$$
\end{itemize} 
Consider the $(c_0,C)$-colourful $2^\ell$-universal tree $\Uc_{(c_0,C)}^\ell$. %with $n\leq 2^\ell$ many leaves. 
\begin{itemize}
    \item if $\ell = 0$ and $C = \emptyset$, then clearly, $\Lc_C^\ell$, defined as $\tpl{}$ uses $0$ bits to label each node in the tree. 
    \item if $\ell = 0$ and $C\neq \emptyset$,  
    then we define $\Lc_C^\ell$ to be the prefix closure of
    $$\bigcup_i \tpl{(\varepsilon\cdot c_{i})\odot \Lc_{C_i}^0}$$ 
    where each $\Lc_{C_i}^0$ is the recursively obtained labelling for $\Uc_{(c_i,C_i)}^0$. Observe that
    no extra bits are used in addition to the bits used by each $\Lc^0_{C_i}$. Since each $\Lc^0_{C_i}$ uses $0$ bits to label their nodes, $\Lc_C^\ell$ also uses $0$ bits to label each node in the tree. 
    \item if $\ell > 0$ and $C\neq \emptyset$ and recall that $$\Uc^\ell_{(c_0,C)} \:=
    \Uc_{(c_0,C)}^{\ell-1}\cdot
\tpl{c_0,\seq{\Uc_{(c_1,C_1)}^\ell,\dots,\Uc_{(c_h,C_h)}^\ell,\left(\bot,\seq{}\right)}}
\cdot\Uc_{(c_0,C)}^{\ell-1}.$$
Let $\Lc_C^\ell$ defined as follows be a labelling of $\Uc_{(c_0,C)}^\ell$, defined as the prefix-closure of 
$$0\cdot\Lc_C^{\ell-1} \,\cup \,\bigcup_i\,(\varepsilon\cdot c_i)\odot \Lc_{C_i}^\ell ,\cup \, (\varepsilon\cdot \bot)\,\cup\,\,1\cdot\Lc_C^{\ell-1}$$
where $\Lc_C^{\ell-1}$ and $\Lc_{C_i}^\ell$ are labellings of $\Uc_{(c_0,C)}^{\ell-1}$ and  $\Uc_{(c_i,C_i)}^{\ell}$ respectively, and use at most $\ell-1$ and $\ell$ bits to encode each of their nodes. Hence $\Lc_C^\ell$ as constructed uses at most $\ell$ bits to encode each node.
\end{itemize}
\end{proof}

%%%%%%-----------------------------------
%%%%%%-----------------------------------
\begin{restatable}{lemma}{nextnodeOne}~\label{lemma:nextnodeOne}
    Given a node in the $\words$-labelled $(c_0,C)$-colourful tree $\Lc_C^\ell$, with at most $2^\ell$ leaves  one can compute the next node larger than a given node in time $O(k\log(k)\ell)$, where $k = |C|$.
\end{restatable}

\begin{proof}\label{prf:nextnodeOne}
    We first introduce, for $a\in\mathbb{N}$, a function
$\nextbit{a}$ that takes a string $\omega$ on $\{0,1\}^*$ with $|\omega|\leq a$ and calculates the smallest $\omega'$ with $|\omega'|\leq a$ that is larger than $\omega$, if it exists (with respect to the ordering on $\words$).

For example, for $a = 3$, the succinct encoding gives us the following order:
$$000<00<001<0<010<01<011<\varepsilon < 100< 10<101<1<110<111 $$

and the $\nextbit{a}$ function gives us exactly this ordering. That is for instance, $\nextbit{3}(0)=010$ and   $\nextbit{3}(011)=\varepsilon$. Additionally, for a newly introduced element $\strtop$, we set $\nextbit{a}(1^a) := \strtop$, i.e. $\nextbit{3}(111)=\strtop$.

Let $\omega \in \{0,1\}^*$ with $|\omega| \leq b$. Then $\nextbit{a}(\omega)$ is computed as follows,
\begin{itemize}
\item If $|\omega| < a$, then $\nextbit{a}(\omega) = \omega10^{a-1}$,
\item If $|\omega|=a$,
    \begin{itemize}
        \item If $\omega = \omega'01^k$ for some $\omega'$ and $k\geq 0$, then  
        $\nextbit{a}(\omega) = \omega'$,
        \item If $\omega = 1^a$, then $\nextbit{a}(\omega) = \strtop$.
    \end{itemize}
\end{itemize}

Next we define our desired function \nexxt{t} that takes a node of $\Lc_C^\ell$ and sends it to the next node that is larger than $t$, and contains colours from the set $C$. If no such node exists, it sends it to $\top$.%\thejaswini{To irmak: can you also modify this to actually compute this above? It is not doing this at the moment. Here, note that this $C$ need not be the set of all colours. } 

We apply the following rules to calculate \nexxt{t} for some node $t = (\omega_1\cdot c_{i_1},\dots,\omega_m\cdot c_{i_m})$:
\begin{itemize}
\item If $c_{i_m} \neq \bot$, then $t$ is not a leaf and therefore, \nexxt{t} is $t$'s smallest child. 
\nexxt{t}$=(\omega_1\cdot c_{i_1}, \ldots, \omega_m\cdot c_{i_m}, 0^b \cdot c)$ where $c$ is the minimum colour in $C \cup \{\bot\} \setminus \{c_{i_1}, \ldots ,c_{i_{m}}\}$ and $b = \ell - \sum_{i= 1}^{m} |\omega_{i}|$.

\item If $c_{i_m} = \bot$, then $t$ is a leaf, therefore \nexxt{t} is the smallest sibling of $t$ that is larger than itself. Therefore, 
$\mynext_C^{\ell}(t)= (\omega_1\cdot c_{i_1}, \ldots, \omega_{m-1}\cdot c_{m-1}, \nextbit{b}(\omega_{m})\cdot c)$ 
where $c$ is the minimum colour in $C \cup \{\bot\} \setminus \{c_{i_1}, \ldots ,c_{i_{m-1}}\}$ and $b = \ell - \sum_{i=1}^{m-1} |\omega_{i}|$.
\end{itemize}

Moreover, for $\omega_j= \strtop$, we have

$$ (w_1\cdot c_{i_1},\dots, w_{j-1}\cdot c_{i_{j-1}},  \strtop \cdot c_{i_j}) = 
\begin{cases}
(\omega_1\cdot c_{i_1},\dots,\omega_{j-1}\cdot c) \quad \text{ if }c_{i_{j-1}}\neq \bot,\\
(\omega_1\cdot c_{i_1},\dots,\nextbit{b}(\omega_{j-1})\cdot c') \text{ if }c_{i_{j-1}} = \bot,\\
\end{cases}$$
 Note that both of these tuples are $(j-1)-$tuples. Here, $c$ is the smallest colour larger than $c_{i_{j-1}}$ in  $C \cup \{\bot\}\setminus \{c_{i_1}, \ldots ,c_{i_{j-2}}\}$, $c'$ is the minimum colour in $C \setminus \{c_{i_1}, \ldots, c_{j-2}\}$ and $b = \ell - \sum_{i=1}^{j-2} |w_{i}|$. 
 
The value $\strtop$ is assigned to the last entry of \nexxt{t} by the application of rules presented above, only when $t$ is the largest of its siblings. In this case, we reassign \nexxt{t} to the smallest sibling of $t$'s parent that is larger than itself, as given above.
Similarly, if $t = (\strtop\cdot \bot)$, then $t = \top$, since $\Lc_C^\ell$ is out of nodes. 

We conclude this detailed computation of $\mynext_C^\ell$ with the observation that the above computation takes only time $O(k\log(k)\ell)$, linear in the length of a node stored.
\end{proof}

%%%%%%%%--------------------------------
%%%%%%%%--------------------------------
\begin{proposition}~\label{prop:nextnodeTwo}
    Given a node $t$ in the $\words$-labelled $(c_0,C)$-colourful tree $\Lc_C^\ell$, with at most $2^\ell$ leaves  and $K\subseteq C$ such that $\colouring(t) \in K$, the next node larger than $t$ such that $\accumulatedcolourset(t)\subseteq K\cup \set{\bot}$ can be found in time $O(k\log(k)\ell)$, where $k = |C\cup c_0|$.
\end{proposition}
\begin{proof}
    For any node $t :=(\omega_1\cdot c_{i_1},\dots,\omega_m\cdot c_{i_m})$ we know $c_{i_m}\in K$.  
    We first find largest position $s$, such that $\accumulatedcolourset((\omega_1\cdot c_{i_1},\dots,\omega_s\cdot c_{i_s})) \subseteq K$.

We then compute the next node $t'$ to $(\omega_1\cdot c_{i_1},\dots,\omega_{s+1}\cdot c_{i_{s+1}})$ which has $\accumulatedcolourset(t')\cap (C\setminus K) = \emptyset$.
But for the tree $\Lc_C^\ell$ constructed, consider the smallest colour $c$  such that $c$ is the smallest colour in $\left(K\cup \{\bot\}\right) \setminus \{c_{i_1},c_{i_2},\dots, c_{i_s}\}$ larger than $c_{i_{s+1}}$.
Observe that the above set is non-empty as such colour exists, as $\bot\in \left(K\cup \{\bot\}\right) \setminus \{c_{i_1},c_{i_2},\dots, c_{i_s}\}$.
Therefore, we only need to return $ (\omega_1\cdot c_{i_1},\dots,\omega_s\cdot c_{i_s}, \omega_{s+1}\cdot c)$, which always exists in the tree constructed. Moreover, it is the smallest node larger than $t$ such that $\accumulatedcolourset(t)\subseteq K\cup \set{\bot}$.

This takes only linear time in the size of the encoding of a node, which is $O(k\log(k)\ell)$.
\end{proof}

%%%%%%%----------------------------------
%%%%%%%----------------------------------
%\mainRabin*

 % \section{Appendix for Section~\ref{sec:fairness}}
 % \input{AppendixD}


\end{document}